\documentclass[letter, 11pt]{article}

\usepackage{natbib}
\usepackage{amssymb,amsmath,xcolor,graphicx,xspace,colortbl,rotating} %
\usepackage[raggedrightboxes]{ragged2e} 
\usepackage{textcomp}
\usepackage{appendix}  
\usepackage{bm}  
\usepackage{boxedminipage}  
\usepackage{color}  
\usepackage{endnotes}  
\usepackage{ragged2e}  
\usepackage[onehalfspacing]{setspace}  
\usepackage{tabulary}  
\usepackage{varioref}  
\usepackage{wrapfig}  
\graphicspath{{Uniqueness_of MFE POST EC1_graphics/}{Uniqueness_of MFE POST EC1_tcache/}{Uniqueness_of MFE POST EC1_gcache/}}
\DeclareGraphicsExtensions{.pdf,.eps,.ps,.png,.jpg,.jpeg}
\usepackage[paper=letterpaper, twoside=false,textwidth=6.3in, textheight=8.7in,left=0.98in, top=0.99in,headheight=0.17in, headsep=0.17in,]{geometry}
\DeclareGraphicsExtensions {.pdf,.svg,.eps,.ps,.png,.jpg,.jpeg}
\newtheorem {theorem}{Theorem}

\newtheorem {assumption}{Assumption}

\newtheorem {corollary}{Corollary}

\newtheorem {definition}{Definition}
\newtheorem {example}{Example}

\newtheorem {lemma}{Lemma}

\newtheorem {proposition}{Proposition}

\newenvironment {proof}[1][Proof]{\noindent \textbf {#1.} }{\ \rule {0.5em}{0.5em}}

\begin{document}
\title{Mean Field Equilibrium: Uniqueness, Existence, and Comparative Statics\protect\footnote{The authors wish to thank Aaron Bodoh-Creed, Ramesh Johari, Bob Wilson, three anonymous referees, and the associate and area editors, as well as seminar participants at Stanford and several conferences for their valuable comments. The second author thanks Joseph and Laurie Lacob for the support during the 2017-2018 and 2018-2019 academic years as a Joseph and Laurie Lacob  Faculty Scholar   at Stanford Graduate School of Business. }}
\author{Bar Light\protect\footnote{Graduate School of Business, Stanford University, Stanford, CA 94305, USA. e-mail: \textsf{barl@stanford.edu}}
~ and Gabriel Y. Weintraub\protect\footnote{  Graduate School of Business, Stanford University, Stanford, CA 94305, USA. e-mail: \textsf{gweintra@stanford.edu }\ }~}
\maketitle

\thispagestyle{empty}

\noindent \noindent \textsc{Abstract}:
\begin{quote}
The standard solution concept for stochastic games is Markov perfect equilibrium (MPE); however, its computation becomes intractable as the number of players increases. Instead, we consider mean field equilibrium (MFE) that has been popularized in the recent literature. MFE takes advantage of averaging effects in models with a large number of players.
We make three main contributions. First, our main result provides conditions that ensure the uniqueness of an MFE. We believe this uniqueness result is the first of its nature in the class of models we study.
Second, we generalize previous MFE existence results. Third,
we provide general comparative statics results. We apply our results to dynamic oligopoly models and to heterogeneous agent macroeconomic models commonly used in previous work in economics and operations.
\end{quote}

\relax
\smallskip \noindent \emph{} 

\newpage 

\section{Introduction} \label{sec:intro}

In this paper we consider a general class of stochastic games in which every player has an individual state that impacts payoffs. Historically, Markov perfect equilibrium (MPE) has been a standard solution concept for this type of stochastic games \citep{MaskinTirole2001}. However, in realistically-sized applications, MPE suffers from two drawbacks. 
First, because in MPE players keep track of the state of every competitor, the state space grows very quickly as the number of players grows, making the
analysis and computation of MPE infeasible in many applications of practical interest. Second, as the number of players increases, it becomes difficult
to believe that players can in fact track the exact state of the other players and optimize their strategies accordingly.

As an alternative,  mean field equilibrium (MFE)  has received extensive attention in the recent literature. In an MFE, each player optimizes her expected discounted payoff, assuming that the
distribution of the other players' states is fixed. Given the  players' strategy, the distribution of the players' states is an invariant distribution of the stochastic process that governs the states' dynamics. 
 As a solution concept for stochastic games, MFE offers several advantages over MPE. First, because
players only condition their strategies on their own state (the competitors' state is assumed to be fixed), MFE is computationally tractable. Second, as several
of the papers we cite below prove, due to averaging effects MFE provides accurate approximations of optimal behavior as the number of players grows. As a result, it provides an appealing behavioral model in games with many players. 

MFE models have many applications in economics, operations research, and optimal control; e.g., studies of anonymous sequential games \citep{jovanovic_1988}, continuous-time
mean field models (\cite{huang2006large} and \cite{lasry2007mean}), dynamic user equilibrium \citep{friesz1993variational}, auction theory (\cite{iyer2014mean}, \cite{balseiro2015repeated}, and \cite{bimpikis2018managing}), dynamic oligopoly models (\cite{weintraub2008markov}
and \cite{adlakha2015equilibria}), heterogeneous agent macro models (\cite{hopenhayn1992entry} and 
 \cite{heathcote2009}), matching
markets (\cite{kanoria2017facilitating} and \cite{arnosti2018managing}), spatial competition \citep{yang2018mean}, and
evolutionary game theory \citep{tembine2009mean}.

We provide three main contributions regarding MFE. First, we provide conditions that ensure the uniqueness of an MFE. This novel result is important because
it implies sharp counterfactual predictions. Second, we generalize previous existence results to a general state space setting. Our existence result includes
the case of a countable state space and a countable number of players, as well as the case of a continuous state space and a continuum of players. In addition,
we provide novel comparative statics results for stochastic games that do not exhibit strategic complementarities. 

We apply our results to well-known dynamic oligopoly models in which individual states represent the firms' ability to compete in the market \citep{doraszelski2007framework}. MFE and the related concept of oblivious equilibrium have previously been used to analyze such models.\footnote{For example, \cite{adlakha2015equilibria} use MFE, which they call \textit{stationary equilibrium}. \cite{adlakha2015equilibria} was motivated by  \cite{hopenhayn1992entry} who introduced the term to study models with infinite numbers of firms. \cite{weintraub2008markov} introduce oblivious equilibrium to study settings with finite numbers of firms.}
 In the models we study, for each firm, being in a larger state is more profitable, while if competitors' states are larger 
 it is less profitable. This structure is quite natural in dynamic models of competition that have been studied in the operations research and economics literature,  and we leverage it to prove our uniqueness result. We provide examples of dynamic investments models of quality, capacity, and advertising, as well as a dynamic reputation model of an online market. We also apply our results to commonly used heterogeneous agent macroeconomic models. 

We now explain our contributions in more detail and compare them to previous work on MFE. 

\textbf{Uniqueness}.
We do not know of any general uniqueness result regarding MFE in discrete-time mean field equilibrium models.\footnote{\cite{lasry2007mean} prove the uniqueness of an MFE in a continuous time setting under a certain monotonicity condition (see also \cite{carmona2018probabilistic}). This monotonicity condition is different and does not hold in the applications studied in the present paper.} Only a few  papers have obtained uniqueness results in specific applications. \cite{hopenhayn1992entry} proves the uniqueness of an MFE in a specific dynamic competition model. \cite{Light2017} proves the uniqueness of an MFE in a Bewley-Aiyagari model under specific conditions on the model's primitives (see a related result in \cite{hu2019unique}). Our main theorem in this paper is a novel result that provides conditions ensuring the uniqueness of an MFE for broader classes of models. Informally, under mild additional technical conditions, we show that if the probability that a player reaches a higher state in the next period is decreasing in the other players' states, and is increasing in the player's own state in the current period, then the MFE is unique (see Theorem \ref{Theorem uniq}). Hence, the conditions reduce the difficulty of showing that a stochastic game has a unique MFE to proving properties of the players' optimal strategies. 

In many applications, one can show that these properties of the optimal strategies arise naturally. For example, in several dynamic models of competition in operations research and economics, a higher firm's state (e.g., the quality of the firm's product or the firm's capacity) implies higher profitability, and the firm can make investments in each period in order to improve its state. In this setting, one can show that a firm invests less when its competitors' states are higher; hence, competitors' higher states induce a lower state for the firm in the next period. In contrast, if the firm's own current state is higher, it induces a higher state in the next period. Another example is heterogeneous agent macro models where each agent solves a consumption-savings problem. The agents' states correspond to their current savings level and current labor productivity. Under certain conditions it can be shown that an agent saves less when the other agents save more.
On the other hand, the agents' next period's savings are increasing in their current savings. 

We apply our uniqueness result to a general class of dynamic oligopoly models and heterogeneous agent macroeconomic models for which MFE has been used to perform counterfactual predictions implied by a policy or system change. In the past, in the absence of this result, previous work mostly focused on a particular MFE selected by a given algorithm, or on one with a specific structure. In the absence of uniqueness, the predictions often depend on the choice of the MFE, and therefore, uniqueness significantly sharpens such counterfactual analysis. We also show that the uniqueness results proved in \cite{hopenhayn1992entry}
and \cite{Light2017} can be obtained using our approach. 

\textbf{Existence.}
Prior literature has considered the existence of equilibria in stochastic games. Some prior work considered the existence of Markov perfect equilibria (MPE)
(see \cite{doraszelski2010computable} and \cite{he2017stationary}).
\cite{adlakha2015equilibria} prove the existence of an MFE for the case
of a countable and unbounded state space. \cite{acemoglu2012}
consider a closely related notion of equilibrium that is called stationary equilibrium and prove its existence for the case of a compact state space and a specific transition dynamic that is commonly used in economics
(see \cite{stokey1989}). Stationary equilibrium in the sense of \cite{acemoglu2012}
is an MFE where the players' payoff functions depend on the other players' states through an aggregator. Our existence result applies for a general compact state space, more general dependence on the
payoff function, and more general transitions.  In this sense, it is more closely related to the result of \cite{adlakha2013mean}. \cite{adlakha2013mean} prove the
existence of an MFE for the case of a compact state space in stochastic games with strategic complementarities using a lattice-theoretical approach. Instead, we do not assume strategic complementarities and our state space can be any compact separable metric space. 
For our existence result, we assume the standard continuity conditions on model primitives that are assumed in
the papers mentioned above. In addition, we assume that the optimal stationary strategy of the players is single-valued.\protect\footnote{In the dynamic oligopoly models and the heterogeneous agent macro models that we study in Sections \ref{Sec: DOM} and \ref{Sec: HAMM}, previous literature assumes that the players use pure strategies. Motivated by this fact, we focus on pure strategy MFE. In this case, if the optimal stationary strategy of the players is not single-valued then the MFE operator may not be convex-valued. Similar problems arise in proving the existence of a pure-strategy Nash equilibrium.} Concavity conditions on the profit function and the transition function can be imposed in order to ensure that the optimal stationary strategy is indeed single-valued. The main technical difficulty in proving existence
is to prove the weak continuity of the nonlinear MFE operator (see Theorem \ref{Theorem existence}). 

\textbf{Comparative statics.} While some papers contain certain specific results on how equilibria change with the parameters of the model
(for example, see \cite{hopenhayn1992entry} and \cite{aiyagari1994}),
only a few papers have obtained general comparative results in large dynamic economies (see \cite{acemoglu2012}
for a discussion of the difficulties associated with deriving such results). Three notable exceptions are \cite{adlakha2013mean}, \cite{acemoglu2012}, and \cite{acemoglu2018equilibrium}. \cite{adlakha2013mean}
use the techniques for comparing equilibria developed in \cite{milgrom1994comparing} to derive general comparative statics results, and essentially rely on results about the monotonicity of fixed points. The direct application of these results
requires that the MFE operator (see Equation (\ref{MFE OPERATOR}))
be increasing. Our comparative statics results are different because they rely on the uniqueness of an MFE. In particular, the MFE operator is not increasing in our setting (see more details in Section 3). In this sense, our comparative static results are more similar to the results in \cite{acemoglu2012};
however, our model has more general dynamics that include, for example, investment decisions with random outcomes that are typically considered in dynamic oligopoly
models (see Section \ref{Sec: DOM}). Our results are useful because they establish the directional changes of MFE when important model parameters, such as the discount factor and the investment cost, change. 

\section{The Model} \label{Section: model}

In this section we define our general model of a stochastic game and define mean field equilibrium (MFE). The model and the definition of an MFE are similar to \cite{adlakha2013mean} and \cite{adlakha2015equilibria}. 

\subsection{Stochastic Game Model}
In this section we describe our stochastic game model. Differently to standard stochastic games in
the literature (see \cite{shapley_1953}), in our model, every player has
an individual state. Players are coupled through their payoffs and state transition dynamics. A stochastic game has the following elements:~

\textit{Time}. The game is played in discrete time. We index time periods by $t =1 ,2 ,\ldots \;\;\text{.}\;\;$ 

\textit{Players. }There are $m$ players in the game. We use $i$ to denote a particular player. 

\textit{States.} The state of player $i$ at time $t$ is denoted by $x_{i ,t} \in X$ where $X$ is a separable metric space. Typically, we assume that the state space $X$ is in $\mathbb{R}^{n}$ or that $X$ is countable. We denote the state of all players at time $t$ by $\boldsymbol{x}_{\mathbf{t}}$ and the state of all players except player $i$ at time $t$ by $\boldsymbol{x}_{ -\boldsymbol{i} ,\boldsymbol{t}}$. 

\textit{Actions. }The action taken by player $i$ at time $t$ is denoted by $a_{i ,t} \in A$ where $A \subseteq \mathbb{R}^{q}$. We use $\boldsymbol{a}_{\mathbf{t}}$ to denote the action of all players at time $t$. The set of feasible actions for a player in state $x$ is given by $\Gamma  (x) \subseteq A$. 

\textit{States' dynamics. }The state of a player evolves in a Markov fashion. Formally, let
$h_{t} =\{\boldsymbol{x}_{0} ,\boldsymbol{a}_{0} ,\ldots  ,\boldsymbol{x}_{\boldsymbol{t} -1} ,\boldsymbol{a}_{\boldsymbol{t} -1}\}$ denote the \emph{history} up to time $t$. Conditional on $h_{t}$, players' states at time $t$ are independent of each other. This assumption implies that random shocks are idiosyncratic, ruling out aggregate random shocks
that are common to all players. Player $i$'s state $x_{i ,t}$ at time $t$ depends on the past history $h_{t}$ only through the state of player $i$ at time $t -1$, $x_{i ,t -1}$; the states of other players at time $t -1$, $\boldsymbol{x}_{ -\boldsymbol{i} ,\boldsymbol{t} -1}$; and the action taken by player $i$ at time $t -1$, $a_{i ,t -1}$. 

If player $i$'s state at time $t -1$ is $x_{i ,t -1}$, the player takes an action $a_{i ,t -1}$ at time $t -1$, the states of the other players at time $t -1$ are $\boldsymbol{x}_{ -\boldsymbol{i} ,\boldsymbol{t} -1}$, and $\zeta _{i ,t}$ is player $i$'s realized idiosyncratic random shock at time $t$, then player $i$'s next period's state is given by
\begin{equation*}x_{i ,t} =w (x_{i ,t -1} ,a_{i ,t -1} ,\boldsymbol{x}_{ -\boldsymbol{i},\boldsymbol{t} -1} ,\zeta _{i ,t}).
\end{equation*}We assume that
$\zeta $ is a random variable that takes values $\zeta _{j} \in E$ with probability $p_{j}$ for $j =1 ,\ldots  ,n$. $w :X \times A \times X^{m -1} \times E \rightarrow X$ is the transition function. 

\textit{Payoff.} In a given time period, if the state of player $i$ is $x_{i}$, the state of the other players is $\boldsymbol{x}_{ -\boldsymbol{i}}$, and the action taken by player $i$ is $a_{i}$, then the single-period payoff to player $i$ is $\pi  (x_{i} ,a_{i} ,\boldsymbol{x}_{ -\boldsymbol{i}}) \in \mathbb{R}$. In Section \ref{sec:exten} we extend our model to a model in which players are also coupled through actions, that is, the functions $w$ and $\pi$ can also depend on the rivals' current actions.

\textit{Discount factor.} The players discount their future payoff
by a discount factor $0 <\beta  <1$. Thus, a player $i$'s infinite horizon payoff is given by: $\sum _{t =1}^{\infty }\beta ^{t-1} \pi  (x_{i ,t} ,a_{i ,t} ,\boldsymbol{x}_{ -\boldsymbol{i} ,\boldsymbol{t}})$. 

In many games, coupling between players is independent of the identity of the players. This notion of anonymity
captures scenarios where the interaction between players is via aggregate information about the state (see \cite{jovanovic_1988}).
Let $s_{ -i ,t}^{(m)} (y)$ denote the fraction of players excluding player $i$ that have their state as $y$ at time $t$. That is,
\begin{equation*}s_{ -i ,t}^{(m)} (y) =\frac{1}{m -1} \sum _{j \neq i}1_{\{x_{j ,t} =y\}}
\end{equation*}where $1_{D}$ is the indicator function of the set $D$. We refer to $s_{ -i ,t}^{(m)}$ as the population state at time $t$ (from player $i$'s point of view).

\begin{definition}
(Anonymous stochastic game). A stochastic game is called an \textit{\emph{anonymous stochastic game }}if
the payoff function $\pi  (x_{i ,t} ,a_{i ,t} ,\boldsymbol{x}_{ -\boldsymbol{i} ,\boldsymbol{t}})$ and the transition function $w (x_{i ,t} ,a_{i ,t} ,\boldsymbol{x}_{ -\boldsymbol{i} ,\boldsymbol{t}} ,\zeta _{i ,t +1})$ depend on $\boldsymbol{x}_{ -i ,t}$ only through $s_{ -i ,t}^{(m)} $. In an abuse of notation, we write $\pi  (x_{i ,t} ,a_{i ,t} ,s_{ -i ,t}^{(m)})$ for the payoff to player $i$, and $w (x_{i ,t} ,a_{i ,t} ,s_{ -i ,t}^{(m)} ,\zeta _{i ,t +1})$ for the transition function for player $i$. 
\end{definition}

For the remainder of the paper, we focus our attention on anonymous stochastic games. For
ease of notation, we often drop the subscripts $i$ and $t$ and denote a generic transition function by $w (x ,a ,s ,\zeta )$ and a generic payoff function by $\pi (x ,a ,s)$ where $s$ represents the population state of players other than the player under consideration. Anonymity requires that a player's single-period payoff and transition function depend on the states of other players via their empirical distribution over the state space, and not on their specific
identify. In anonymous stochastic games the functional form of the payoff function and transition function are the same, regardless of the number of players
$m$.\protect\footnote{Our results also
generalize for models in which the primitives depend on the number of players $m$ like in the study of oblivious equilibria \citep{weintraub2008markov}).} In that sense, we often interpret the profit function $\pi (x ,a ,s)$ as representing a limiting regime in which the
number of players is infinite. 

We now provide a simple model of capacity competition that illustrates some of the notation presented above.  This is one of the dynamic competition models that we study in Section \ref{Subsec:Quality ladder}.

\begin{example} \label{ex:capacity}  
Our example is based on the capacity competition models of \cite{besanko2004capacity} and \cite{besanko2010lumpy}.
We consider an industry with homogeneous products, where each firm's state variable determines its production capacity. If the firm's state is $x$, then its capacity is $\bar{q} (x)$. In each period, each firm takes a costly action to improve its  capacity in the next period. Further, in each period, firms compete in a capacity-constrained quantity setting game. The inverse demand function is given by $P (Q)$, where $Q$ represents the total industry output. For simplicity, we assume that the marginal costs of all the firms are equal to zero.
Given the total quantity produced by its competitors $Q_{-i}$, the profit maximization problem for firm $i$ is given by $\underset{0 \leq q_{i} \leq \bar{q} (x_{i})}{\max }P (q_{i} +Q_{ -i}) q_{i}$. 

In general, one could solve for the equilibrium of the capacity-constrained static quantity game played by firms, and these static equilibrium actions would determine the single-period profits. However, we focus on the 
 limiting regime with a large number of firms with out market power, that is, firms take $Q$ as fixed. In this case, each firm produces at full capacity and  the limiting profit function is given by:
\begin{equation*} \pi(x,a,s) =P \left (\int_X \bar{q} (y) s (d y)\right ) \bar{q} (x) - d a,
\end{equation*} 
where  $a$ is the firm's investment and $d$ is the unit investment cost (see also \cite{ifrach2016framework}). 
The  next period's state depends on the amount of investment,  the current state, and a random shock. For example, assuming that the state depreciates at rate $\delta$, a possible transition function is given by:
$$w(x,a,s, \zeta) = ((1-\delta)x + k(a)) \zeta ,$$
where $k$ is an increasing function that determines the impact of the firm's investment and $\zeta$ represents uncertainty in the investment process.
\end{example}

Now, we let $\mathcal{P} (X)$ be the set of all possible population states on $X$, that is $\mathcal{P} (X)$ is the set of all probability measures on $X$. We endow $\mathcal{P} (X)$ with the weak topology. Since $\mathcal{P} (X)$ is metrizable, the weak topology on $\mathcal{P} (X)$ is determined by weak convergence (for details see \cite{aliprantis2006infinite}). We
say that $s_{n} \in \mathcal{P} (X)$ converges weakly to $s \in \mathcal{P} (X)$ if for all bounded and continuous functions $f :X \rightarrow \mathbb{R}$ we have
\begin{equation*}\underset{n \rightarrow \infty }{\lim }\int _{X}f (x) s_{n} (d x) =\int _{X}f (x) s (d x).
\end{equation*}

{\em For the rest of the paper}, we assume the following conditions on the primitives of the model: 
\begin{assumption} 
\label{Assumption 0}(i) $\pi $ is bounded and (jointly) continuous. $w$ is continuous.\protect\footnote{  Recall that
we endow $\mathcal{P} (X)$ with the weak topology. } 

(ii) $X$ is compact. 

(iii) The correspondence $\Gamma  :X \rightarrow 2^{A}$ is compact-valued and continuous.\protect\footnote{
By continuous we mean both upper hemicontinuous and lower hemicontinuous. } 
\end{assumption}

\subsection{Extensions To The Basic Model} \label{sec:exten}

We note two extensions that can be important in applications for which we can extend our results.

First, in our basic mean field model, we assume that the players are coupled through their states: both the transition function and the payoff function of each player depend on the states of all other players. We note that even in this setting, a player's payoff function can depend on rivals' actions as long as these actions do not affect the evolution of their own state  nor the evolution of the population state. For instance, the players' payoff functions can depend on the {\em static} pricing or quantity decisions of the other players. In Section \ref{Subsec:Quality ladder} we study models in which the firms' (static) actions affect other players' current payoffs but do not affect the evolution of future states.

In certain models of interest such as learning-by-doing and dynamic advertising, however, players' states are coupled through the {\em dynamic} actions, $a_{i,t}$. That is, the actions of other players, $\boldsymbol{a}_{ -\boldsymbol{i} ,\boldsymbol{t} }$, affect a player's transition function and payoff function.
For these cases, we consider a model where the transition function and the payoff function of each player depend on both the states and the actions of all other players. The model is like our original model except that now the probability measure $s$ describes the joint distribution of players over actions and states and not only over states, that is, $s \in \mathcal{P}(X \times A)$. Thus, the transition function $w(x,a,s,\zeta)$ and the payoff function $\pi (x,a,s)$ depend on the joint distribution over state-action pairs $s \in \mathcal{P}(X \times A)$. 

All the results in the paper can be extended to this setting where the population state is a measure on $\mathcal{P}(X \times A)$ (see Section \ref{Sec: Coupling th} in the Appendix for more details). The monotonicity conditions that are needed in order to prove the uniqueness of an MFE in the case that the population is a measure on $\mathcal{P}(X \times A)$ are similar to the conditions that are needed in the case that the population is a measure on $\mathcal{P}(X)$. In Section \ref{sec:advertising} we prove the uniqueness of an MFE for a dynamic advertising model where the players' payoff functions depend on the other players' actions (advertising expenditures), and thus, the population state is a measure on $\mathcal{P}(X \times A)$.

Our second extension relaxes the assumption on our base model that players are ex-ante homogeneous.
To consider players that may be ex-ante heterogeneous with different model primitives, we extend our model to a setting in which each player has a fixed type through out the time horizon that is drawn from a finite set. Then, the payoff function and transition function can depend on this type. We show that all our results hold in this more general setting (see Section \ref{Sec: Ex-ante h} for more details). In particular, we show that if the conditions that we use in order to prove our results hold for every type, then the results are valid for the model with ex-ante heterogeneous players.

\subsection{Mean Field Equilibrium}
In Markov perfect equilibrium (MPE), players' strategies are functions of the population state. However, MPE quickly becomes intractable as the
number of players grows, because the number of possible population states becomes too large. Instead, in a game with a large number of players, we might
expect that idiosyncratic  fluctuations of players' states ``average out'', and hence the actual population state remains roughly constant over time. Because the effect
of other players on a single player's payoff and transition function is only via the population state, it is intuitive that, as the number of players increases,
a single player's effect on the outcome of the game is negligible. Based on this intuition, related schemes for approximating Markov perfect equilibrium
(MPE) have been proposed in different application domains via a solution concept we call mean field equilibrium (MFE). 

Informally, an MFE is a strategy for the players and a population state such that: (1) Each player optimizes her expected discounted payoff assuming that this population
state is fixed; and (2) Given the players' strategy, the fixed population state is an invariant distribution of the states' dynamics.
The interpretation is that a single player conjectures the population state to be $s$. Therefore, in determining her future expected payoff stream, a player considers a payoff function and a transition function evaluated at the fixed population state $s$. In MFE, the conjectured $s$ is the correct one given the strategies being played. MFE alleviates 
the complexity of MPE, because in the former 
the population state is fixed, while in the latter players keep track of the exact evolution of the population state. We refer the reader to the papers cited in Section \ref{sec:intro} for a more detailed motivation and rigorous justifications for using MFE.

Let $X^{t}:=\underbrace{X \times \ldots  \times X}_{t~ \mathrm{t} \mathrm{i} \mathrm{m} \mathrm{e} \mathrm{s}}$. For a fixed population state, a nonrandomized pure strategy $\sigma $ is a sequence of (Borel) measurable functions $(\sigma _{1} ,\sigma _{2} ,\ldots  ,)$ such that $\sigma _{t} :X^{t} \rightarrow A$ and $\sigma _{t} (x_{1} ,\ldots  ,x_{t}) \in \Gamma  (x_{t})$ for all $t \in \mathbb{N}$. That is, a strategy $\sigma $ assigns a feasible action to every finite string of states. Note that a single player's strategy depends only on her own
history of states and does not depend on the population state. This strategy is called an \emph{oblivious} strategy (see \cite{weintraub2008markov}
and \cite{adlakha2015equilibria}). 

For each initial state $x \in X$ and long run average population state $s \in \mathcal{P} (X)$, a strategy $\sigma $ induces a probability measure over the space $X^{\mathbb{N}}$, describing the evolution of a player's state.\protect\footnote{The probability measure on $X^{\mathbb{N}}$ is uniquely defined (see for example \cite{bertsekas1978stochastic}).} We denote the expectation with respect to that probability measure by $\mathbb{E}_{\sigma }$, and the associated states-actions stochastic process by $\{x (t) ,a (t)\}_{t =1}^{\infty }$. 

When a player uses a strategy $\sigma $, the population state is fixed at $s \in \mathcal{P} (X)$, and the initial state is $x \in X$, then the player's expected present discounted value is
\begin{equation*}V_{\sigma } (x ,s) =\mathbb{E}_{\sigma } \left (\sum _{t =1}^{\infty }\beta ^{t -1} \pi  (x (t) ,a (t) ,s) \right ).
\end{equation*}
Denote
\begin{equation*}V(x,s) = \underset{\sigma } {\sup }\;V_{\sigma} (x,s).
\end{equation*}
That is, $V(x,s)$ is the maximal expected payoff that the player can achieve when the initial state is $x$ and the population state is fixed at $s \in \mathcal{P} (X)$. We call $V$ the \emph{value function} and a strategy $\sigma$ attaining it \emph{optimal}. 

Standard dynamic programming arguments (see \cite{bertsekas1978stochastic}) show that the value function satisfies the Bellman equation:
\begin{equation*}V (x ,s) =\underset{a \in \Gamma  (x)}{\max } ~\pi  (x ,a ,s) +\beta  \sum _{j =1}^{n}p_{j} V (w (x ,a ,s ,\zeta _{j}) ,s).
\end{equation*}
Under Assumption 1, there exists an optimal stationary Markov strategy (see Lemma \ref{Lemma 1} in the Appendix). Let $G (x ,s)$ be the optimal stationary strategy correspondence, i.e.,
\begin{equation*}G (x,s) =\underset{a \in \Gamma  (x)}{\ensuremath{\operatorname*{argmax}}}\; \pi  (x ,a ,s) +\beta  \sum _{j =1}^{n}p_{j} V (w (x ,a ,s ,\zeta _{j}) ,s). 
\end{equation*}
Let $\mathcal{B}(X)$ be the Borel $\sigma$-algebra on $X$. For a strategy $g \in G$ and a fixed population state $s \in \mathcal{P} (X)$, the probability that player $i$'s next period's state will lie in a set $B \in \mathcal{B} (X)$, given that her current state is $x \in X$ and she takes the action $a =g (x ,s)$, is:
\begin{equation*}Q_{g} (x ,s ,B) =\Pr  (w (x ,g (x ,s) ,s ,\zeta ) \in B).
\end{equation*}
Now suppose that the population state is $s$, and all players use a stationary strategy $g \in G$. Because of averaging effects, we expect that if the number of players is large, then the long run population state should in
fact be an invariant distribution of the Markov kernel $Q_{g}$ on $X$ that describes the dynamics of an individual player. 

We can now define an MFE. In an MFE, every player
conjectures that $s$ is the fixed long run population state and plays according to a stationary strategy $g$. On the other hand, if every agent plays according to $g$ when the population state is $s$, then the long run population state of all players, $s$, should constitute an invariant distribution of $Q_{g}$.

\begin{definition}
\label{Def MFE}A stationary strategy $g$ and a population state $s \in \mathcal{P} (X)$ constitute an MFE if the following two conditions hold: 

1. Optimality: $g$ is optimal given $s$, i.e., $g (x ,s) \in G (x ,s)$. 

2. Consistency: $s$ is an invariant distribution of $Q_{g}$. That is,
\begin{equation*}s (B) =\int _{X}Q_{g} (x ,s ,B) s (d x).
\end{equation*}for all $B \in \mathcal{B} (X)$, where we take Lebesgue integral with respect to the measure $s$. 
\end{definition}

Under Assumption \ref{Assumption 0} it can be shown that $G (x ,s)$ is nonempty, compact-valued and upper hemicontinuous. The proof is a standard application of the maximum theorem. We provide the proof for completeness (see Lemma \ref{Lemma 1}). In Theorem \ref{Theorem existence} we prove the existence of a population state that satisfies the consistency requirement in Definition \ref{Def MFE}.

\section{Main Results}
In this section we present our main results. In Section 3.1 we provide conditions that ensure the uniqueness of an MFE. In Section 3.2 we prove the existence of an MFE. In Section 3.3 we provide conditions that ensure unambiguous comparative statics results regarding MFE. 

\subsection{The Uniqueness of an MFE } \label{sec:unique}
In this section we present our uniqueness result. 


We recall that a stationary strategy-population state pair $(g ,s)$ is an MFE if and only if $g$ is optimal and $s$ is a fixed point of the operator $\Phi  :\mathcal{P} (X) \rightarrow \mathcal{P} (X)$ defined by
\begin{equation}\Phi  s (B) =\int _{X}Q_{g} (x ,s ,B) s (d x), \label{MFE OPERATOR}
\end{equation}for all $B \in \mathcal{B} (X)$.

We prove uniqueness by showing that the operator $\Phi $ has a unique fixed point. In order to prove uniqueness we will assume that $G$ is single-valued. For the rest of the section we will assume that $g \in G$ is the unique selection from the optimal strategy correspondence $G$. In the next section we provide conditions that ensure that $G$ is indeed single-valued (see Lemma \ref{Lemma concave}). $G$ being single-valued and Theorem \ref{Theorem existence} (see Section 3.2) imply
that $\Phi $ has at least one fixed point. In Theorem \ref{Theorem uniq} we will show that
under certain conditions the operator $\Phi $ has at most one fixed point. 

 We omit the reference to $g$ in $Q_{g} (x ,s ,B)$, i.e., we write $Q (x ,s ,B)$ instead of $Q_{g} (x ,s ,B)$. 
Since the Markov kernel $Q$ depends on $s$, it is complicated to work directly with the operator $\Phi $. Thus, to prove the uniqueness of an MFE and to prove our comparative statics results, we introduce an auxiliary operator that is easier to work with. For each $s \in \mathcal{P} (X)$, define the operator $M_{s} :\mathcal{P} (X) \rightarrow \mathcal{P} (X)\;$by
\begin{equation*}M_{s} \theta  (B) =\int _{X}Q (x ,s ,B) \theta  (d x).
\end{equation*}
We introduce the following useful definition. 

\begin{definition} \label{def:Xerg}
We say that $Q$ is $X$-ergodic if the following two conditions hold: 

(i) For any $s \in \mathcal{P} (X)$, the operator $M_{s}$ has a unique fixed point $\mu _{s}$. 

(ii) $M_{s}^{n} \theta $ converges weakly to $\mu _{s}$ for any probability measure $\theta  \in \mathcal{P} (X)$. 
\end{definition}
Note that $s$ is an MFE if and only if $\mu _{s} =s$ is a fixed point of the operator $M_{s}$. 
$X$-ergodicity means that for every population state $s \in \mathcal{P} (X)$ the players' long-run state is independent of the initial state. The $X$-ergodicity of $Q$ can be established using standard results from the theory of Markov chains in general state spaces (see \cite{meyn2012markov}). When $Q$ is increasing in $x$, which we assume in order to prove the uniqueness of an MFE (see Assumption \ref{Assumption uniq}), then the $X$-ergodicity of $Q$ can be established using results from the theory of monotone Markov chains. These results usually require a splitting condition (see \cite{bhattacharya1988asymptotics} and \cite{hopenhayn1992}) that typically holds in applications of interest. Specifically, in Sections \ref{Sec: DOM} and \ref{Sec: HAMM} we show that $X$-ergodicity holds in important classes of dynamic models. 

We now introduce other notation and definitions that are helpful in proving uniqueness. 
We assume that $X$ is endowed with a closed partial order $ \geq $. In the important case $X =\mathbb{R}^{n}$, $x ,y \in X$ we write $x \geq y$ if $x_{i} \geq y_{i}$ for each $i =1  ,.. , n$. Let $S \subseteq X$. We say that a function $f :S \rightarrow \mathbb{R}$ is increasing if $f (y) \geq f (x)$ whenever $y \geq x$ and we say that $f$ is strictly increasing if $f(y)>f(x)$ whenever $y>x$. 

For $s_{1} ,s_{2} \in \mathcal{P} (X)$ we say that $s_{1}$ stochastically dominates $s_{2}$ and we write $s_{1}  \succeq  _{SD}s_{2}$ if for every increasing function $f :X \rightarrow \mathbb{R}$ we have
\begin{equation*}\int _{X}f (x) s_{1} (dx) \geq \int _{X}f (x) s_{2}(dx),
\end{equation*}when the integrals exist. We say that $B \in \mathcal{B} (X)$ is an upper set if $x_{1} \in B$ and $x_{2} \geq x_{1}$ imply $x_{2} \in B$. Recall from   \cite{kamae1977stochastic}
that $s_{1}  \succeq  _{SD}s_{2}$ if and only if for every upper set $B$ we have $s_{1} (B) \geq s_{2} (B)$. 

In addition, for the rest of the section we will assume that there exists a binary relation $ \succeq $ on $\mathcal{P} (X)$, such that $s_{2} \sim s_{1}$ (i.e., $s_{2} \succeq s_{1}$ and $s_{1} \succeq s_{2})$) implies $\pi  (x ,a ,s_{1}) =\pi  (x ,a ,s_{2})$ for all $(x ,a) \in X \times A$ and $w (x ,a ,s_{1} ,\zeta ) =w (x ,a ,s_{2} ,\zeta )$ for all $(x ,a ,\zeta ) \in X \times A \times E$. 

Note that such binary relation always exists, for example one can take $s_{2} \sim s_{1} \Leftrightarrow s_{2} = s_{1}$.  For our uniqueness result we will further require that the binary relation $ \succeq $ on $\mathcal{P} (X)$ is complete, that is, for all $s_{1} ,s_{2} \in \mathcal{P} (X)$ we either have $s_{1} \succeq s_{2}$ or $s_{2} \succeq s_{1}$.   In many applications (see Section \ref{Sec: DOM} and Section \ref{Sec: HAMM})
there exists a function $H :\mathcal{P} (X) \rightarrow \mathbb{R}$ such that $\tilde{\pi}  (x ,a ,H (s)) =\pi(x ,a ,s)$ and $\tilde{w} (x ,a ,H (s) ,\zeta )=w(x,a ,s ,\zeta )$, where $H$ is continuous and increasing with respect to the stochastic dominance order $ \succeq  _{SD}$. In this case, a natural complete order $ \succeq $ on $\mathcal{P} (X)$ arises by defining $s_{1} \succeq s_{2}$ if and only if $H (s_{1}) \geq H (s_{2})$. Below, we also discuss the case of a non-complete order. We say that $ \succeq $ \emph{agrees with} $ \succeq  _{SD}$ if for any $s_{1} ,s_{2} \in \mathcal{P} (X)$, $s_{1}  \succeq  _{SD}s_{2}$ implies $s_{1} \succeq s_{2}$.

We  say that $Q$ is increasing in $x$ if for each $s \in \mathcal{P} (X)$, we have $Q (x_{2} ,s , \cdot ) \succeq  _{SD}Q (x_{1} ,s , \cdot )$ whenever $x_{2} \geq x_{1}$. In addition, we say that $Q$ is decreasing in $s$ if for each $x \in X$, we have $Q (x ,s_{1} , \cdot ) \succeq  _{SD}Q (x ,s_{2} , \cdot )$ whenever $s_{2} \succeq s_{1}$. We now state the main theorem of the paper. We show that if $Q$ is $X$-ergodic, $Q$ is increasing in $x$ and decreasing in $s$, and $ \succeq $ is complete and agrees with $ \succeq  _{SD}$, then if an MFE exists, it is unique. 

Intuitively, $Q$ decreasing in $s$ implies that the probability that a player will move to a higher state in the next period is decreasing in the current period's
population state. If there are two MFEs, $s_{2}$ and $s_{1}$, such that $s_{2} \succeq s_{1}$ (i.e., $s_{2}$ is ``higher'' than $s_{1}$), then the probability of moving to a higher state under $s_{2}$ is lower than under $s_{1}$, which is not consistent with $s_{2} \succeq s_{1}$, with the definition of an MFE,  and the fact that $ \succeq $ agrees with $ \succeq  _{SD}$.\footnote{In some models, the condition that $Q$ is decreasing in $s$ follows from the fact that the policy function $g$ is decreasing in the population state $s$ (see Section \ref{Sec: DOM}).  \cite{xu2013supermarket} prove the uniqueness of an equilibrium in a supermarket mean field game under a similar monotonicity condition on the policy function. Their setting is different from ours because the players do not have individual states nor they dynamically optimize.}

\begin{assumption}
\label{Assumption uniq} 

(i) $Q$ is $X$-ergodic. $Q$ is increasing in $x$ and decreasing in $s$. 

(ii) $ \succeq $ agrees with $ \succeq  _{SD}$. 

(iii) $G$ is single-valued. 
\end{assumption}

\begin{theorem}
\label{Theorem uniq}Suppose that Assumption \ref{Assumption uniq}
holds. If the binary relation $ \succeq $ is complete, then if an MFE exists, it is unique. 
\end{theorem}

\begin{proof}
Let $\theta _{1} ,\theta _{2} \in \mathcal{P} (X)$ and assume that $\theta _{1}  \succeq  _{SD}\theta _{2}$. Let $B$ be an upper set and let $s_{1} ,s_{2}$ be two MFEs such that $s_{2} \succeq s_{1}$. We have
\begin{align*}M_{s_{2}} \theta _{2} (B) &  =\int _{X}Q (x ,s_{2} ,B) \theta _{2} (dx) \\
 &  \leq \int _{X}Q (x ,s_{1} ,B) \theta _{2} (dx) \\
 &  \leq \int _{X}Q (x ,s_{1} ,B) \theta _{1} (dx) \\
 &  =M_{s_{1}} \theta _{1} (B). \end{align*}
 Thus, for any upper set $B$ we have $M_{s_{2}} \theta _{2} (B) \leq M_{s_{1}} \theta _{1} (B)$ which implies that $M_{s_{1}} \theta _{1}  \succeq  _{SD}M_{s_{2}} \theta _{2}$. The first inequality follows from the fact that $Q (x ,s ,B)$ is decreasing in $s$ for an upper set $B$ and all $x$. The second inequality follows from the fact that $\theta _{1}  \succeq  _{SD}\theta _{2}$ and $Q (x ,s ,B)$ is increasing in $x$ for an upper set $B$ and any $s$. 

We conclude that $M_{s_{1}}^{n} \theta _{1}  \succeq  _{SD}M_{s_{2}}^{n} \theta _{2}$ for all $n \in \mathbb{N}$. $Q$ being $X$-ergodic implies that $M_{s_{i}}^{n} \theta _{i}$ converges weakly\protect\footnote{  Recall that $\mu _{s}$ is the unique fixed point of $M_{s}$ and that $s$ is an MFE if and only if $\mu _{s} =s$. } to $\mu _{s_{i}} =s_{i}$. Since $ \succeq  _{SD}$ is closed under weak convergence (see \cite{kamae1977stochastic}), we have $s_{1}  \succeq  _{SD}s_{2}$. 

We conclude that if $s_{1}$ and $s_{2}$ are two MFEs such that $s_{2} \succeq s_{1}$, then $s_{1}  \succeq  _{SD}s_{2}$. Since $ \succeq $ agrees with $ \succeq  _{SD}$, we have $s_{1} \succeq s_{2}$. That is, $s_{1} \sim s_{2}$, which implies that $\pi  (x ,a ,s_{1}) =\pi  (x ,a ,s_{2})$ and $w (x ,a ,s_{1} ,\zeta ) =w (x ,a ,s_{2} ,\zeta )$. Thus, under $s_{1}$ the players play according to the same strategy as under $s_{2}$ (i.e., $g (x ,s_{1}) =g (x ,s_{2})$ for all $x \in X$). We conclude that $Q (x ,s_{1} ,B) =Q (x ,s_{2} ,B)$ for all $x \in X$ and $B \in \mathcal{B} (X)$. $X$-ergodicity of $Q$ implies that $M_{s_{1}}$ and $M_{s_{2}}$ have a unique fixed point. Thus, $\mu _{s_{1}} =\mu _{s_{2}}$, i.e., $s_{1} =s_{2}$. Similarly, we can show that $s_{1} \succeq s_{2}$ implies that $s_{1} =s_{2}$. 

Since $ \succeq $ is complete if $s_{1}$ and $s_{2}$ are two MFEs we have $s_{2} \succeq s_{1}$ or $s_{1} \succeq s_{2}$. Thus, we proved that if $s_{1}$ and $s_{2}$ are two MFEs then $s_{1} =s_{2}$. We conclude that if an MFE exists, it is unique. 
\end{proof}

The assumptions on $Q$ in Theorem \ref{Theorem uniq} involve assumptions on the optimal strategy $g$. Thus, these assumptions are not over the primitives of the model. In Section \ref{Sec: DOM}
we introduce conditions on the primitives of dynamic oligopoly models that guarantee the uniqueness of an MFE. In particular, we show that the monotonicity
conditions over $Q$ arise naturally in important classes of these models. In Section \ref{Sec: HAMM} we apply our result to prove the uniqueness of an MFE in heterogeneous agent macro models.

In some applications the assumption that the binary relation $ \succeq $ is complete is restrictive. In the case that $ \succeq $ is not complete and Assumption \ref{Assumption uniq} holds, the following Corollary shows that the MFEs are not comparable by the binary relation $ \succeq $. 
This Corollary can be used to derive properties on the MFE when there are multiple MFEs. For example, suppose that there exist two functions $H_{i}:\mathcal{P}(X) \rightarrow\mathbb{R}$, $i =1 ,2$ such that $\tilde{\pi} (x ,a ,H_{1}(s))=\pi (x ,a ,s)$ and $\tilde{w}(x ,a ,H_{2}(s),\zeta )=w (x,a ,s ,\zeta )$, where $H_{i}$ is continuous and increasing with respect to the stochastic dominance order $ \succeq  _{SD}$. We can define an order $ \succeq $ on $\mathcal{P} (X)$ by defining $s_{1} \succeq s_{2}$ if $H_{1}(s_{1}) \geq H_{1}(s_{2})$ and $H_{2}(s_{1}) \geq H_{2}(s_{2})$. Clearly, this may not be a complete order. The following Corollary provides conditions that imply that if $s_{1}$ and $s_{2}$ are two MFEs, then it cannot be the case that $H_{1}(s_{1}) > H_{1}(s_{2})$ and $H_{2}(s_{1}) > H_{2}(s_{2})$. We write $s_{1} \succ s_{2}$ if $s_{1} \succeq s_{2}$ and $s_{2}\nsucceq s_{1}$.

\begin{corollary}
Suppose that Assumption \ref{Assumption uniq} holds. If $s_{1}$ and $s_{2}$ are two MFEs then $s_{1} \nsucc s_{2}$ and $s_{2} \nsucc s_{1}$. 
\end{corollary}

\begin{proof}
Suppose, in contradiction, that $s_{2} \succ s_{1}$. The argument in the proof of Theorem \ref{Theorem uniq} implies that $s_{1}  \succeq  _{SD}s_{2}$. Since $ \succeq $ agrees with $ \succeq  _{SD}$, we have $s_{1} \succeq s_{2}$, which is a contradiction. We conclude that $s_{2} \nsucc s_{1}$. Similarly, we can show that $s_{1} \nsucc s_{2}$. 
\end{proof}

When the state space $X$ is given by the product space $X=X_{1} \times X_{2}$ where $X_{1}$ and $X_{2}$ are separable metric spaces, a modification of our uniqueness result can be applied to prove the uniqueness of an MFE under slightly different conditions than the conditions of Assumption \ref{Assumption uniq}. 

Assumption \ref{Assumption uniq} requires that $Q$ be increasing in $x$ on $X$. However, when $X=X_{1} \times X_{2}$, and $X_{i}$ is endowed with the closed partial order $ \geq_{i}$, it is enough to assume that $Q$ is increasing in $x_{i}$ on $X_{i}$ for some $i=1,2$ to prove that the MFE is unique. We say that $Q$ is increasing in $x_{1}$ if for all functions $f :X_{1} \times X_{2} \rightarrow \mathbb{R}$ that are increasing in $x_{1}$ on $X_{1}$, for all $s \in \mathcal{P} (X)$, and for all $x_{2} \in X_{2}$, the function 
\begin{equation}\label{Eq:increasing1}
    \int_{X} f(y_{1} ,y_{2})Q((x_{1} ,x_{2}) ,s ,d(y_{1} ,y_{2}))
    \end{equation}
    is increasing in $x_{1}$. 
Similarly, $Q$ is decreasing in $s$ with respect to $x_{1}$ if for all functions $f :X_{1} \times X_{2} \rightarrow \mathbb{R}$ that are increasing in $x_{1}$on $X_{1}$ and for all $x \in X$ the function in (\ref{Eq:increasing1}) is decreasing in $s$.
In Sections \ref{Sec:dynamic rep} and \ref{Sec: HAMM} we show the usefulness of Theorem \ref{Corr:product space}. We establish the uniqueness of an MFE for dynamic reputation models and heterogeneous agent macro models by proving that $Q$ is increasing in $x_{i}$ for some $i=1,2$. In these models it is not necessarily true that $Q$ is increasing in $x$ on $X$, so Theorem \ref{Theorem uniq} cannot be applied directly. The Appendix contains the proofs not presented in the main text. 

\begin{theorem}\label{Corr:product space}
Suppose that $X=X_{1} \times X_{2}$. Suppose that Assumption \ref{Assumption uniq} holds,  apart from the condition that $Q$ is increasing in $x$ and decreasing in $s$. Suppose that $Q$ is increasing in $x_{i}$ and decreasing in $s$ with respect to $x_{i}$ for some $i=1,2$. If the binary relation $ \succeq $ is complete, then if an MFE exists, it is unique. 
\end{theorem}

\subsection{The Existence of an MFE } \label{sec:exist}
In this section we study the existence of an MFE. We show that if $G$ is single-valued, then the operator $\Phi $ defined in Equation (\ref{MFE OPERATOR}) has a fixed point and thus, there exists
an MFE.

\begin{theorem}
\label{Theorem existence}Assume that $G\;$is single-valued. There exists a mean field equilibrium. 
\end{theorem}

Note that we do not impose Assumption \ref{Assumption uniq} for this result. 
Also note that $X$ can be any compact separable metric space in the proof of Theorem \ref{Theorem existence},
so the existence result holds for the important cases of finite state spaces, countable state spaces,  and $X \subseteq \mathbb{R}^{n}$. In addition, the proof of existence does not depend on the number of players in the game; the number of players in the game can be
finite, countable or uncountable. Finally, we note that we do not require $X$-ergodicity (see Definition \ref{def:Xerg}) to show existence; instead we use compactness and continuity (see Assumption \ref{Assumption 0}). The main challenge to prove existence 
is to prove the weak continuity of the nonlinear MFE operator. To do so, we leverage a generalized version of the bounded convergence theorem by \cite{serfozo1982convergence}. 

We now provide conditions over the model primitives that guarantee that $G$ is single-valued when $X$ is a convex set in $\mathbb{R}^{n}$. Similar conditions have been used in dynamic oligopoly models.\footnote{For similar results in a countable state space
setting see \cite{adlakha2015equilibria} and \cite{doraszelski2010computable}).}

\begin{assumption} \label{Assumption exist}
\label{Assumption single valued}Suppose that $X \subseteq \mathbb{R}^{n}$ and is convex. 

(i) Assume that $\pi  (x ,a ,s)$ is concave in $(x ,a)$, strictly concave in $a$ and increasing in $x$ for each $s \in \mathcal{P} (X)$. 

(ii) Assume that $w$ is increasing in $x$ and concave in $(x ,a)$ for each $\zeta  \in E$. 

(iii) $\Gamma  (x)$ is convex-valued and increasing in the sense that $x_{2} \geq x_{1}$ implies $\Gamma  (x_{2}) \supseteq \Gamma  (x_{1})$. 
\end{assumption}

The following Lemma shows that the preceding conditions on the primitives of the model ensure that
$G$ is single-valued.

\begin{lemma}
\label{Lemma concave}Suppose that Assumption \ref{Assumption single valued}
holds. Then $G$ is single-valued. 
\end{lemma}

The previous results can be summarized by the following Corollary that imposes
conditions over the primitives of the model which guarantee the existence of an MFE.

\begin{corollary}
Suppose that Assumption \ref{Assumption single valued} holds. Then, there exists
an MFE. 
\end{corollary}

\subsection{Comparative Statics}
In this section we derive comparative statics results. Let $(I , \succeq _{I})$ be a partially ordered set that influences the players' optimal decisions. We denote
a generic element in $I$ by $e$. For example, $e$ can be the discount factor, a parameter that influences the players' payoff functions, or a parameter that influences the players' transition dynamics. Throughout this section we slightly
abuse notation and when the parameter $e$ influences the players' optimal decisions we add it as a parameter. For instance, we write $Q (x ,s ,e , \cdot )$ instead of $Q (x ,s , \cdot )$. We say that $Q$ is increasing in $e$ if $Q (x ,s ,e_{2} , \cdot ) \succeq  _{SD}Q (x ,s ,e_{1} , \cdot )$ for all $x$, $s$, and all $e_{2} ,e_{1} \in I$ such that $e_{2}  \succeq _{I}e_{1}$. We prove that under the assumptions of Theorem \ref{Theorem uniq}, if $Q$ is increasing in $e$ then $e_{2}  \succeq _{I}e_{1}$ implies that the unique MFE under $e_{2}$ is higher than the unique MFE under $e_{1}$ with respect to $ \succeq $. 

\cite{adlakha2013mean} derive comparative statics results for MFE in the case that $Q$ is increasing in $s$, $x$ and $e$. They prove that $e_{2}  \succeq _{I}e_{1}$ implies $s (e_{2})  \succeq  _{SD}s (e_{1})$ where $s (e)$ is the maximal MFE with respect to $ \succeq  _{SD}$ under $e$. \cite{adlakha2013mean} use the techniques to
compare equilibria developed in \cite{milgrom1994comparing} (see also \cite{topkis2011supermodularity}).
We note that under the assumptions of Theorem \ref{Theorem uniq}, $Q$ is increasing in $x$ but decreasing in $s$. Thus, the results in \cite{adlakha2013mean} do not apply to our setting. However, with the help of the uniqueness of an MFE, we derive a general comparative statics result.

\begin{theorem}
\label{Theorem comp}Let $(I, \succeq _{I})$ be a partial order. Assume that $Q$ is increasing in $e$ on $I$. Then, under the assumptions of Theorem \ref{Theorem uniq}, the unique MFE $s (e)$ is increasing in the following sense: $e_{2}  \succeq _{I}e_{1}$ implies $s (e_{2}) \succeq s (e_{1})$. 
\end{theorem}

The same result can be shown with a similar argument under the assumptions of Theorem \ref{Corr:product space}. We omit the details for sake of brevity.  We note that our comparative statics result is with respect to the order $ \succeq $ and not with respect to the usual stochastic dominance order. The machinery mentioned in the paragraph above is not directly applicable in our models, and without it we believe that comparative statics results with respect to the usual stochastic dominance order are much harder to obtain. We discuss the usefulness of our comparative static result with respect to the order $\succeq$ in the context of dynamic oligopoly models below.

\section{\label{Sec: DOM}Dynamic Oligopoly Models}

In this section we study various dynamic models of competition  or dynamic oligopoly models that capture a wide range of phenomena in economics and operations research.\footnote{Even though we study models with potentially large numbers of firms, we keep the name dynamic oligopoly to be consistent with previous literature in which MFE or its variants have been used to approximate oligopolistic behavior (for example, see \cite{qi2013impact}, \cite{adlakha2015equilibria}, and \cite{onishi2016quantity}).} We leverage our results to provide conditions under which a broad class of dynamic oligopoly models admit a unique MFE. We also provide comparative statics results. 

More specifically, we show that under concavity assumptions and a natural substitutability condition, the MFE is unique. The substitutability condition requires that the firms' profit function has decreasing differences in each firm's own state and the states of the other firms. 
This condition implies that the marginal profit of a firm (with respect to its own state) is decreasing in the other firms' states. It arises naturally in many dynamic oligopoly models. In Section \ref{Subsec:Quality ladder} we consider well studied capacity competition and quality ladder models. In Section \ref{sec:advertising} we consider a dynamic advertising model. In Section \ref{Sec:dynamic rep} we introduce a dynamic reputation model of an online market.  In all of these models, it holds that the firms' actions are higher when their own state is higher and the firms' actions are lower when the competitors' states (or the competitors' actions) are higher. These are essentially the conditions that imply the uniqueness of an MFE for dynamic oligopoly models.

\subsection{Capacity Competition and Quality Ladder Models}\label{Subsec:Quality ladder}

In this section we consider dynamic capacity competition models and dynamic quality ladder models which have received significant attention in the recent operations research and economics literature. In these models, firms' states correspond to a variable that affects their profits. For example, the state can be the firm's capacity or the quality of the firm's product. Per-period
profits are based on a static competition game that depends on the heterogeneous firms' state variables. Firms take actions in order to improve their individual state  over time. 

We now describe the models we consider.

\textit{States.} The state
of firm $i$ at time $t$ is denoted by $x_{i ,t} \in X$ where $X \subseteq \mathbb{R}_{ +}$ and is convex. 

\textit{Actions.} At each time $t$, firm $i$ invests $a_{i ,t} \in A =[0 ,\bar{a}]$ to improve its state. The investment changes the firm's state in a stochastic fashion. 

\textit{States' dynamics.} A firm's state evolves in a Markov fashion. Let $0 <\delta  <1$ be the depreciation rate. If firm $i$'s state at time $t -1$ is $x_{i ,t -1}$, the firm takes an action $a_{i ,t -1}$ at time $t -1$, and $\zeta _{i ,t}$ is firm $i$'s realized idiosyncratic random shock at time $t$, then firm $i$'s state in the next period is given by:
\begin{equation*}x_{i ,t} =((1 -\delta ) x_{i ,t -1} +k (a_{i ,t -1})) \zeta _{i ,t}
\end{equation*} where $k :A \rightarrow \mathbb{R}$  is typically an increasing function that determines the impact of investment $a$.
We assume that $\zeta$ takes positive values $0 <\zeta _{1} <\ldots  <\zeta _{n}$, where $\zeta _{1} <1$, $\zeta _{n} >1$, $p_{1} ,p_{n} >0$. That is, there exists a positive probability for a bad shock $\zeta _{1}$ and a positive probability for a good shock $\zeta _{n}$. In each period, the firm's state is naturally depreciating at rate $\delta$, but the firm can make investments in order to improve it. Further, the outcome of depreciation and investment is subject to an idiosyncratic random shock ($\zeta$) that, for example, could capture uncertainty in R\&D or a marketing campaign. 
Related dynamics have been used in previous literature. 
Further, our uniqueness result for capacity competition and quality ladder models holds under other states' dynamics. For example, we could also assume additive dynamics $x_{i,t} =(1 -\delta )x_{i,t-1}+k(a_{i ,t -1}) +\zeta _{i ,t}$.\footnote{For our results to hold we need to impose some constraints on these additive dynamics so that the state space remains compact. We can also  assume an exogenous bound on the state as in Section \ref{Sec:dynamic rep}. We believe that our results also hold if we drop the assumption that $X$ is compact, under some additional conditions over model primitives that ensure some form of ``decreasing returns to larger states'' (see \cite{adlakha2015equilibria}).} We make the following assumption over the dynamics that we discuss later before Theorem \ref{th:unique_olig}.
\begin{assumption}
\label{Assumption DOP:0}

(i) $k(a)$ is strictly concave, continuously differentiable, strictly increasing and $k(0)>0$.\protect\footnote{
The differentiability assumptions can be relaxed. We assume differentiability of $u$ and $k$ in order to simplify the proof of Theorem \ref{Unique DOP}.}

(ii) $(1 -\delta ) \zeta _{n} <1$. 
\end{assumption} 
\textit{Payoff.} The cost of a unit of investment is $d >0$.\footnote{The investment  cost could be a convex function, but linearity simplifies the comparative static results in the parameter $d$.} We assume there is a single-period profit function $u (x ,s)$ derived from a static game. When a firm invests $a \in A$, the firm's state is $x \in X$, and the population state is $s \in \mathcal{P} (X)$, then the firm's single-period payoff function is given by $\pi  (x ,a ,s) =u (x ,s)-da$.  

We assume that there exists a complete and transitive binary relation $ \succeq $ on $\mathcal{P} (X)$ such that $s_{1} \sim s_{2}$ implies that $u (x ,s_{1}) =u (x ,s_{2})$ for all $s_{1} ,s_{2} \in \mathcal{P} (X)$ and $x \in X$. Furthermore, we assume that $ \succeq $ agrees with $ \succeq _{SD}$ (cf. Section \ref{sec:unique}). 

To prove the uniqueness of an MFE for capacity competition and quality ladder models, we introduce the following conditions on the primitives of the model. It is simple to verify that both of the dynamic oligopoly models introduced in the examples below satisfy these assumptions. We believe the conditions are quite natural, and thus other commonly used dynamic oligopoly models may satisfy them as well.

Recall that a function $f (x ,s)$ is said to have decreasing differences in $(x ,s)$ on $X \times S$ if for all $x_{2} \geq x_{1}$ and $s_{2} \succeq s_{1}$ we have
$f (x_{2} ,s_{2}) -f (x_{1} ,s_{2}) \leq f (x_{2} ,s_{1}) -f (x_{1} ,s_{1}).$ $f$ is said to have increasing differences if $ -f$ has decreasing differences.

\begin{assumption}
\label{Assumption DOP}
 $u (x ,s)$ is jointly continuous. Further, it is concave and continuously differentiable in $x$, for each $s \in \mathcal{P} (X)$. In addition,  $u (x ,s)$ has decreasing differences in $(x,s)$. 

\end{assumption} 

We now provide two classic examples of profit functions $u (x ,s)$ that are commonly used in the literature. For these examples, we explicitly define the binary relation $ \succeq $.

The first one is the capacity competition model described in Example \ref{ex:capacity}. Recall that 
if  the firm's state is $x$, then its capacity is $\bar{q} (x)$. We assume that $\bar{q}$ is an increasing, continuously differentiable, concave, and bounded function. We also assume that the inverse demand function $P(\cdot)$ is decreasing and continuous. In this model, \begin{equation*}u (x ,s) =P \left (\int_X \bar{q} (y) s (d y)\right ) \bar{q} (x).
\end{equation*}
For the capacity competition model, we define $s_{2} \succeq s_{1}$ if and only if $\int \bar{q} (y) s_{2} (d y) \geq \int \bar{q} (y) s_{1} (d y)$. Since $\bar{q}$ is an increasing function, $ \succeq $ agrees with $ \succeq  _{SD}$. It can be verified that $u$ satisfies the conditions of Assumption \ref{Assumption DOP}.

Our second example is a classic quality ladder model, where individual states represent the quality of a firm's product (see, e.g., \cite{pakes1994computing} and \cite{ericson1995markov}). Consider a price competition under a logit demand system. There are $N$ consumers in the market. The utility of consumer $j$ from consuming the good produced by firm $i$ at period $t$ is given by
\begin{equation*}u_{i j t} =\theta _{1} \ln  (x_{it} +1) +\theta _{2} \ln  (Y-p_{it}) +v_{ijt},
\end{equation*}
where $\theta _{1}<1 ,\theta _{2} >0$, $p_{it}$ is the price of the good produced by firm $i$, $Y$ is the consumer's income, $x_{it}$ is the quality of the good produced by firm $i$, and $\{v_{i j t}\}_{i ,j ,t}$ are i.i.d Gumbel random variables that represent unobserved characteristics for each consumer-good pair. 

There are $m$ firms in the market and the marginal production cost is constant and the same across firms. There is a unique Nash equilibrium in pure strategies of the pricing
game (see \cite{caplin1991aggregation}). These equilibrium static prices determine the single-period profits. Now, the limiting profit function that we focus on can be obtained from the asymptotic regime in which the number of consumers $N$ and the number of firms $m$ grow to infinity at the same rate. The limiting profit function corresponds to a logit model of monopolistic competition given by: 
\begin{equation*}u (x ,s) =\frac{\tilde{c} (x +1)^{\theta _{1}}}{\int _{X}(y +1)^{\theta _{1}} s (d y)}
\end{equation*} (see \cite{besanko1990logit}). 
$\tilde{c}$ is a constant that depends on the limiting equilibrium price, the marginal production cost, the consumer's income, and $\theta _{2}$. For the quality ladder model, we define $s_{2} \succeq s_{1}$ if and only if $\int (y +1)^{\theta _{1}} s_{2} (d x) \geq \int (y +1)^{\theta _{1}} s_{1} (d y)$. It is easy to see that $\succeq $ agrees with $\succeq  _{SD}$. It can also be verified that $u$ satisfies the conditions of Assumption \ref{Assumption DOP}.

The proof of our uniqueness result for the capacity competition and quality ladder models consists of showing that Assumptions \ref{Assumption DOP:0} and \ref{Assumption DOP} imply Assumptions \ref{Assumption 0} and \ref{Assumption uniq}, and that $ \succeq $ is a complete order. These are the conditions we use to show the existence of a unique MFE in Sections \ref{sec:unique} and \ref{sec:exist}. 

Specifically, similarly to Lemma \ref{Lemma concave}, one can show that the concavity assumptions in Assumptions \ref{Assumption DOP:0} and \ref{Assumption DOP} imply that $G$ is single-valued. The assumption that $k(0)>0$ (see condition (i) in Assumption \ref{Assumption DOP:0}) is used to prevent the pathological case that the Dirac measure on the point $0$ is an invariant distribution of $M_{s}$ which could violate $X$-ergodicity (see Section \ref{sec:unique}). In addition, condition (ii) in Assumption \ref{Assumption DOP:0} is used to control the growth of firms, so that one can show that the state space remains compact. We believe our results hold with a milder version of this assumption. With this, the only remaining assumption that we need to show in order to prove the uniqueness of an MFE  for our capacity competition and quality ladder models is Assumption \ref{Assumption uniq}(i). For this, we use the fact that the profit function has decreasing differences in the state $x$ and the population state $s$. This  implies that firms invest less when the population state is higher (see Lemma \ref{Lemma single-valued}).
 We use this fact to show the desired monotonicity of $Q$. 

Our main result for dynamic capacity competition and dynamic quality ladder models is the following:

\begin{theorem} \label{th:unique_olig}
\label{Unique DOP}Suppose that Assumptions \ref{Assumption DOP:0} and \ref{Assumption DOP}
hold. Then there exists a unique MFE for the capacity competition and quality ladder models. 
\end{theorem}

 Under Assumptions \ref{Assumption DOP:0} and \ref{Assumption DOP} we can also derive comparative statics results for our capacity competition and  quality ladder models. In particular, we show that an increase in the cost of a unit of investment decreases the unique MFE population state. Note that an increase in the investment cost decreases firms incentives to invest. However, a lower population state incentivizes the firms to invest more. As a consequence, our model does not have the properties of a supermodular game (e.g., \cite{topkis1979equilibrium}). Despite this, relying on the uniqueness
of an MFE and on Theorem \ref{Theorem comp} we are able to show that in fact the unique MFE decreases when the cost of a unit of investment increases. 

We also derive comparative statics results regarding a change in a parameter that influences the profit function and a change in the discount factor. We show that if there is a parameter $c$ such that the marginal profit of the firms is decreasing in that parameter, then the unique MFE decreases in the parameter $c$. For example, in the quality ladder model, as the marginal cost of production goes up, the unique MFE decreases. In the capacity competition model, as the potential market size increases, the MFE increases. In addition, we show that an increase in the discount factor increases the unique MFE. 

We note that all of our comparative statics results are with respect to the order $\preceq$ and not with respect to the usual stochastic dominance order as one would typically obtain using supermodularity arguments (e.g., \cite{adlakha2013mean}). We believe that these results provide helpful information because the order $\preceq$ relates to the single-period profit function, and therefore, MFE can be ordered in terms of firms' payoffs. Further, $\preceq$ typically orders a variable of economic interest, such as the average capacity level in the capacity competition model or the average quality level in the quality ladder model.

\begin{theorem}
\label{Theorem DOP comp}Suppose that Assumptions \ref{Assumption DOP:0} and \ref{Assumption DOP}
hold. We denote by $s (e)$ the unique MFE when the parameter that influences the firms' decisions is $e$. 

(i) If the cost of a unit of investment increases, then the unique MFE decreases, i.e., $d_{2} \leq d_{1}$ implies $s (d_{2}) \succeq s (d_{1})$. 

(ii) Let $c \in I \subseteq \mathbb{R}$ be a parameter that influences the firms' profit function. If the profit function $u (x ,s ,c)$ has decreasing differences in $(x ,c)$, then the unique MFE decreases in $c$, i.e., $c_{1} \geq c_{2}$ implies $s (c_{2}) \succeq s (c_{1})$. 

(iii) Assume that $u (x ,s)$ is increasing in $x$. If the discount factor $\beta $ increases, then the unique MFE $s (\beta )$ increases, i.e., $\beta _{2} \geq \beta _{1}$ implies $s (\beta _{2}) \succeq s (\beta _{1})$. 
\end{theorem}

\subsection{Dynamic Advertising Competition Models}\label{sec:advertising} 

In this section we consider dynamic advertising competition models. In these models, firms' states correspond to customer goodwill or market share. 
In each period, the firms decide on their advertising expenditures $a$. The probability that the next period's customer goodwill is higher increases when the firms spend more on advertising. The firms' payoff functions depend on their own spending on advertising, on their own state, on the other firms' states, and on the other firms' spending on advertising. Thus, a firm's payoff function depends on the other firms' {\em dynamic} actions (in Sections \ref{sec:exten} and \ref{Sec: Coupling th} we extend the model and the results presented in Sections 2 and 3 to the case in which each player's payoff function depends on the other players' actions).    
Variants of dynamic models with this structure have been studied in the operations research literature in contexts other than advertising (for example, see \cite{hall2000customer}). We now describe our specific model.\footnote{Our model is a mean field version of the dynamic advertising model presented in \cite{heyman2004stochastic} and in Section 4.3 in \cite{olsen2014markov}).}

\textit{States.} The state of firm $i$ at time $t$ is denoted by $x_{i ,t} \in X$ where $X =\mathbb{R}_{ +}$. The state of a firm $x_{i ,t} \in X$ represents the customer goodwill.

\textit{Actions.} At each time $t$, firm $i$ chooses an amount of money to spend on advertising $a_{i ,t} \in A =[1 ,\bar{a}]$ where $\bar{a} >1$. 

\textit{States' dynamics.} When the firm spends more on advertising, the customer goodwill increases. The customer goodwill depreciates over time at rate $0< \delta <1$. If firm $i$'s state at time $t -1$ is $x_{i ,t -1}$, the firm takes an action $a_{i ,t -1}$ at time $t -1$, and $\zeta _{i ,t}$ is firm $i$'s realized idiosyncratic random shock at time $t$, then firm $i$'s state in the next period is given by
\begin{equation*}x_{i ,t} =(1 -\delta )(x_{i ,t -1} +a_{i ,t -1}) \zeta _{i ,t}.
\end{equation*}
We assume that $\zeta $ takes positive values $0 <\zeta _{1} <\ldots  <\zeta _{n}$. To ensure compactness we also assume that $(1-\delta)\zeta _{n} < 1$ (see Section \ref{Subsec:Quality ladder}). We slightly modify the transition dynamics from Section \ref{Subsec:Quality ladder} to remain consistent with the models used in the papers that motivate this section.

\textit{Payoff.} When a firm chooses to spend $a \in A$ on advertising, the firm's state is $x \in X$, and the population action-state profile is $s \in \mathcal{P} (X \times A)$, then the firm's single-period payoff function is given by 
\begin{equation*}\pi (x ,a ,s) = r\frac{(x+a)^{\gamma_{1}}}{(\int(x'+a')s(d(x',a')))^{\gamma_{2}}} - a
\end{equation*}
where $\frac{(x+a)^{\gamma_{1}}}{(\int(x'+a')s(d(x',a')))^{\gamma_{2}}}$ is the expected demand, $r>0$ is the price, and $0< \gamma_{1} < 1$, $0< \gamma_{2} < 1$ are parameters. 
The expected demand is increasing in the firm's current advertising expenditure and in the firm's current state, and is decreasing in the other firms' advertising expenditures and the other firms' states.

 We define a complete binary relation $ \succeq $ on $\mathcal{P} (X \times A)$,  by $s_{1} \succeq s_{2}$ if and only if $(\int(x'+a')s_{1}(d(x',a')))^{\gamma_{2}} \geq  (\int(x'+a')s_{2}(d(x',a')))^{\gamma_{2}}$. Clearly, $ \succeq $ agrees with $ \succeq  _{SD}$ (see Section \ref{sec:unique}). We can also derive comparative statics results for the dynamic advertising model. For example, using similar arguments to the arguments in Section \ref{Subsec:Quality ladder} we can show that when the discount factor $\beta$ increases, then the unique MFE increases in the following sense: if $\beta_{2} > \beta_{1}$, then $s(\beta _{2}) \succeq s(\beta _{1})$ where $s(\beta)$ is the unique MFE under discount factor $\beta$. We also show that the unique MFE increases when the market price $r$ increases.

\begin{theorem}\label{Theorem:dynamic advertising}
(i) The dynamic advertising competition model has a unique MFE. 

(ii) Let $s(\beta)$ be the unique MFE under the discount factor $\beta$. Then $\beta_{2} > \beta_{1}$ implies $s(\beta _{2}) \succeq s(\beta _{1})$. 

(iii) Let $s(r)$ be the unique MFE under the price $r$. Then $r_{2} > r_{1}$ implies $s(r _{2}) \succeq s(r_{1})$. 
\end{theorem}

\subsection{A Dynamic Reputation Model}\label{Sec:dynamic rep}

In this section we consider a dynamic reputation model. Motivated by the proliferation of online markets, reputation models and the design of reputation systems have recently been widely studied in the operations and management science  literature.\footnote{For example, see \cite{dellarocas2003digitization},  \cite{aperjis2010optimal}, \cite{bolton2013engineering},  \cite{papanastasiou2017crowdsourcing}, and \cite{besbes2018information}.} These systems can mitigate the mistrust between buyers and sellers participating in the marketplace (see \cite{tadelis2016reputation}). Further, online markets typically consist of many small sellers, and therefore, assuming an MFE limit is natural.

We study a dynamic reputation model in which sellers improve their reputation level over time. The state of each seller consists of the average review given to her in the past history and the total number of reviews she has received.\footnote{Typically, review systems report simple averages; the number of reviews may also be relevant as it may signal more sales and more experience from a seller.} In each period, each seller receives a review from buyers.\protect\footnote{
This assumption is made only for simplicity. We can also assume that reviews arrive according to a Poisson process.} A seller's ranking is a simple average of her past reviews. Sellers invest in their products in order to improve their reviews over time. For example, Airbnb hosts can invest in cleaning their apartments, and sellers on eBay can invest in their packaging. Higher investments increase the probability of receiving a good review. Sellers' payoffs depend on their rankings and on the number of reviews they receive as well as on the other sellers' rankings and number of reviews. Each seller's payoff function increases in her ranking and in her number of reviews and decreases in the other sellers' rankings and number of reviews. This can capture, for example, the fact that a seller with a higher ranking can charge a higher price or garner more demand. 

The dynamic reputation model we consider in this section assumes that sellers arrive and depart over time. We make this modeling choice because of its realistic appeal, and to ensure that the number of reviews does not tend to infinity. Because we study a stationary setting, we assume that the sellers' rates of arrival and departure balance, so that the market size remains constant over time (in expectation). After each review, a seller departs the market and never returns with probability $1 -\beta $ where $0 <\beta  <1$. For each seller $i$ that departs, a new seller immediately arrives. We assign the new seller the same label $i$, and a $0$ ranking, and $0$ reviews. Under this assumption, it is straightforward to show that the seller's decision problem is the same stationary, infinite horizon, expected discounted reward maximization problem that we introduced in Section \ref{Section: model}, where the discount factor is the probability of remaining in the market.\footnote{For example, \cite{iyer2014mean} provide a similar regenerative model of arrivals and departures.}

We now describe the dynamic reputation model in more detail. 

\textit{States.} The state of seller $i$ at time $t$ is denoted by $x_{i ,t} =(x_{i ,t ,1} ,x_{i ,t ,2}) \in X_{1} \times X_{2} =X$.  $x_{i ,t ,1}$ represents seller $i$'s average numerical review rating up to time $t$. We call $x_{i ,t ,1}$ seller $i$'s ranking at period $t$. $x_{i ,t ,2}$ represents the number of reviews seller $i$ has received up to period $t$. 

\textit{Actions.} At each time $t$, seller $i$ chooses an action $a_{i ,t} \in A =[0 ,\bar{a}]$ in order to improve her ranking. The action changes the seller's state in a stochastic fashion.

\textit{States' dynamics.} If seller $i$'s state at time $t-1$ is $x_{i ,t -1}$, the seller takes an action $a_{i,t -1}$ at time $t -1$, and $\zeta _{i ,t}$ is seller $i$'s realized idiosyncratic random shock at time $t$, then seller $i$'s state in the next period is given by:
\begin{equation*}x_{i ,t} =\left (\min \left (\frac{x_{i ,t-1 ,2}}{1 +x_{i ,t-1 ,2}}x_{i ,t-1 ,1} +\frac{1}{1 +x_{i ,t-1 ,2}}(k(a)+\zeta _{i ,t}) ,M_{1}\right ) ,\min \left (x_{i ,t-1 ,2} +1 ,M_{2}\right )\right ),
\end{equation*}
 where $k :A \rightarrow \mathbb{R}$ is a strictly increasing and strictly concave function that determines the impact of the seller's investment on the next period's review. The next period's numerical review, $k(a) + \zeta$, is assumed to be non-negative.\footnote{In order to simplify the analysis and preserve Assumption \ref{Assumption 0}, we assume that the numerical value of a review $k(a) + \zeta$ can be any non-negative number and not a discrete number. In a model where $k(a) + \zeta$ is discrete our results still hold as long as the optimal strategy is single-valued.}  $M_{1}>0$ is the upper bound on the sellers' ranking and $M_{2}>0$ is the upper bound on the sellers' number of reviews. The latter are useful to keep the state space compact.
The first term in the dynamics represents the simple average of the numerical reviews received so far, while the second term represents the total number of reviews. Similarly to the previous models, the random shocks represent uncertainty in the review process. 

\textit{Payoff.} The cost of a unit of investment is $d>0$. When the seller's ranking is $x_{1}$, the seller's number of reviews is $x_{2}$, the seller chooses an action $a \in A$, and the population state is $s \in \mathcal{P} (X)$, then  the seller's single-period payoff is given by 
\begin{equation*}\pi (x ,a ,s) =\frac{\nu(x_{1} ,x_{2})}{\int \nu(x_{1} ,x_{2})s(d(x_{1} ,x_{2}))} -da
\end{equation*}
where $\nu$ is increasing in $x_1$ and $x_2$, concave, continuously differentiable in $x_{1}$, and positive. The functional form resembles the logit model studied in Section \ref{Subsec:Quality ladder}. 

The cost of a unit of investment can be seen as a lever that a platform may impact by design. In particular, a platform can reduce the cost of a unit of investment for the sellers by introducing tools to improve the buyers' experience of  using the sellers' products.   For example, an e-commerce platform could help facilitating logistics for its sellers, and a rental sharing platform could help its hosts connecting cleaning services.

 We define a complete and transitive binary relation $ \succeq $ on $\mathcal{P} (X)$ by $s_{1} \succeq s_{2}$ if and only if $\int \nu (x_{1} ,x_{2})s_{1}(d(x_{1} ,x_{2})) \geq \int \nu (x_{1} ,x_{2})s_{2}(d(x_{1} ,x_{2}))$. It is easy to see that $ \succeq $ agrees with $ \succeq  _{SD}$ (see Section \ref{sec:unique}).

We  use Theorem \ref{Corr:product space} to prove that the dynamic reputation model admits a unique MFE.\footnote{For this model we are able to show the monotonicity of the kernel $Q$ with respect to $x_1$ but not with respect to $x_2$.} We also show that when the platform reduces the cost of a unit of investment then the MFE increases. 
 
\begin{theorem}
\label{THM: REPUTATION}(i) The dynamic reputation model has a unique MFE.

(ii) Let $s (d)$ be the unique MFE under the unit of investment cost $d$. Then $d_{2} \geq d_{1}$ implies $s(d_{2}) \preceq s(d_{1})$. 
\end{theorem}

\section{\label{Sec: HAMM}Heterogeneous Agent Macroeconomic Models}
In this section we consider heterogeneous agent macro models. In these models, there is a continuum of agents facing idiosyncratic risks only (and no aggregate risks). The
heterogeneous agents make decisions given certain market prices (in \cite{aiyagari1994},
for example, the market prices are the interest rate and the wage rate). The market prices are determined by the aggregate decisions of all the agents in the market. We consider a setting similar to the one presented in \cite{acemoglu2012}.
We note that this setting encompasses many important models in the economics literature. Examples include Bewley-Aiyagari models (see \cite{bewley1986stationary}, and \cite{aiyagari1994}), and models of industry equilibrium (see \cite{hopenhayn1992entry}).
While \cite{acemoglu2012} derive important existence and comparative statics
results for these models, to the best of our knowledge there are no general uniqueness results. In this Section we show that if the agents' strategy is decreasing in the aggregator (in the sense of \cite{acemoglu2012}), there
exists a unique equilibrium. 

We now describe our specific model. 

\textit{States.} The state of player $i$ at time $t$ is denoted by $x_{i ,t} =(x_{i ,t ,1} ,x_{i ,t ,2}) \in X_{1} \times X_{2} =X$ where $X_{1} \subseteq \mathbb{R}$ and $X_{2} \subseteq \mathbb{R}^{n -1}$. For example, in Bewley models $x_{i ,t ,1}$ typically represents agent $i$'s savings at period $t$ and $x_{2}$ represents agent $i$'s income or labor productivity at period $t$ (in this case $n=2$). 

\textit{Actions.} At each time $t$, player $i$ chooses an action $a_{i ,t} \in \Gamma  (x_{i ,t}) \subset \mathbb{R}$. 

\textit{States' dynamics.} The state of a player evolves in a Markovian
fashion. If player $i$'s state at time $t -1$ is $x_{i ,t -1}$, player $i$ takes an action $a_{i ,t -1}$ at time $t -1$, and $\zeta _{i ,t}$ is player $i$'s realized idiosyncratic random shock at time $t$, then player $i$'s state in the next period is given by
\begin{equation*}(x_{i ,t ,1} ,x_{i ,t ,2}) =(a_{i ,t -1} ,m (x_{i ,t -1 ,2} ,\zeta _{i ,t})),
\end{equation*}
where $m :X_{2} \times E \rightarrow X_{2}$ is a continuous function. For example, in Bewley models, in each period agents choose how much to save for future consumption and how much to consume in the current period. The agents' labor productivity evolves exogenously and the labor productivity function $m$ determines the next period's labor productivity given the current labor productivity. So if an agent chooses to save $a$, $\zeta$ is the realized random shock, and her current labor productivity is $x_{2}$, then the agent's next period state (savings-labor productivity pair) is given by $(a,m (x_{2},\zeta))$. 

\textit{Payoff.} As in \cite{acemoglu2012},
we assume that the payoff function depends on the population state through an aggregator. That is, if the population state is $s$, then the aggregator is given by $H (s)$ where $H :\mathcal{P} (X) \rightarrow \mathbb{R}$ is a continuous function. If the aggregator is $H (s)$, the player's state is $x \in X$, and the player takes an action $a \in \Gamma  (x)$, then the player's single-period payoff function is given by $\tilde{\pi}(x ,a ,H (s))$. 

We define a complete and transitive binary relation $ \succeq $ on $\mathcal{P} (X)$ by $s_{1} \succeq s_{2}$ if and only if $H (s_{1}) \geq H (s_{2})$. We assume that $ \succeq $ agrees with $ \succeq  _{SD}$. This assumption holds in most of the heterogeneous agent macro models, where $H$ is usually assumed to be increasing with respect to first order stochastic dominance (see \cite{acemoglu2012}).

Note that under the states' dynamics defined above, and assuming that $g (x ,s)=\tilde{g}(x,H(s))$ is the optimal stationary strategy, the transition kernel $Q$ is given by
\begin{equation*}Q (x_{1} ,x_{2} ,s ,B_{1} \times B_{2}) =1_{B_{1}}(\tilde{g} (x_{1} ,x_{2} ,H (s)) \sum _{j}p_{j} 1_{B_{2}} (m (x_{2} ,\zeta _{j})),
\end{equation*}
where $B_{1} \times B_{2} \in \mathcal{B} (X_{1} \times X_{2})$. 

We show that the model has a unique MFE if the optimal strategy is decreasing in the aggregator, i.e., if $H (s_{2}) \geq H (s_{1})$ implies $\tilde{g} (x_{1} ,x_{2} ,H (s_{2})) \leq \tilde{g} (x_{1} ,x_{2} ,H (s_{1}))$, $Q$ is $X$-ergodic, and $\tilde{g}$ is increasing in $x_{1}$. We note that we cannot apply Theorem \ref{Theorem uniq} to this model, since in most applications the optimal stationary strategy $\tilde{g}$ is not increasing in $x_{2}$, and thus $Q$ may not be increasing in $x_{2}$. However, in most applications (for example, all the applications discussed in \cite{acemoglu2012})
$\tilde{g}$ is increasing in $x_{1}$. Thus, we can use Theorem \ref{Corr:product space} to show that the heterogeneous agent macro model has a unique MFE under the conditions stated above.\protect\footnote{
Note that an MFE is usually called a stationary equilibrium in the economics literature (e.g., \cite{acemoglu2012}).}

\begin{corollary}
\label{Hetero Macro uniq} Assume that $G$ is single-valued, $Q$ is $X$-ergodic, and $\tilde{g}$ is increasing in $x_{1}$ and decreasing in the aggregator. Then the heterogeneous agent macro model has a unique MFE.
\end{corollary}

In most applications, the payoff function $\tilde{\pi} $ has increasing differences in $(x_{1} ,a)$ which ensures that $\tilde{g}$ is increasing in $x_{1}$. The condition that $Q$ is $X$-ergodic also usually holds in applications. For example, \cite{aiyagari1994}  proves that $Q$ is $X$-ergodic in his model. Thus, in many applications, in order to ensure uniqueness, one only needs to check that $\tilde{g}$ is decreasing in the aggregator. In the next section we illustrate this  in a Bewley-type model introduced in \cite{aiyagari1994}.

\noindent \textbf{A Bewley-Aiyagari Model}. 
Bewley models are widely studied and used in the modern macroeconomics literature (for a survey see \cite{heathcote2009}). As previously mentioned, in Bewley models agents receive a state-dependent income in each period and they solve an infinite horizon consumption-savings problem; that is, the agents must decide how much to save and how much to consume in each period. The agents can transfer assets from one period to another only by investing in a risk-free bond, and have some borrowing limit. \cite{aiyagari1994} extends the Bewley model to a general equilibrium model with production. We now describe the model of \cite{aiyagari1994} in the setting of a mean field game.

In a Bewley-Aiyagari model, $x_{1}$ represents the agents' savings and $x_{2}$ represents the agents' labor productivity. $m(x_{2},\zeta )$ represents the labor productivity function. That is, if the current labor productivity is $x_{2}$ then the next period's labor productivity is given by $m (x_{2} ,\zeta _{j})$ with probability $p_{j}$. If the agents' labor productivity is $x_{2}$ then their income is given by $wx_{2}$ where $w >0$ is the wage rate. The agents' savings rate of return is $R >0$. 

In each period $t$, the agents choose their next period's savings level $a \in \Gamma  (x_{1} ,x_{2})$  where $\Gamma  (x_{1} ,x_{2})=[ -\underline{b} ,\min \{R x_{1} +w x_{2},\bar{b}\}]$, and consume $c =R x_{1} +w x_{2} -a$. That is, the agents' savings are lower than their cash-on-hand $Rx_{1} + wx_{2}$ and higher than the borrowing constraint $\underline{b} \geq 0$. $\bar{b}$ is an upper bound on savings that ensures compactness. 

The wage rate and the interest rate are determined in general equilibrium. There is a representative firm with a production function $F(K,N)$ that is homogeneous of degree one. $N$ is the labor supply and $K$ is the capital. We assume that $F$ is twice continuously differentiable, strictly concave, and strictly increasing. The first order conditions of the firm's maximization problem yield\footnote{The firm's maximization problem is given by $\max_{K,N} F(K,N) - (R-1+\delta)K - wN$. For more details see, for example, \cite{acemoglu2012} and  \cite{Light2017}.} $F_{k}(K ,N) =R +\delta  -1$ and $F_{N}(K ,N) =w$ where $\delta  >0$ is the depreciation rate and $F_{i}(K ,N)$ denotes the partial derivative of $F$ with respect to $i =K ,N$. A standard argument\footnote{Since $F$ is homogeneous of degree one we have $F(K,1)=KF_{K}(K,1) + F_{N}(K,1)$. Using the first order conditions we have  $f(K)=Kf'(K)+w$.} shows that $R=f'(K)-\delta +1$ and $w =f (K) -f^{ \prime } (K) K$ where $F(K,1)=f(K)$.  

 In equilibrium we have $\int _{X}x_{1} s (d (x_{1} ,x_{2})) =K$ where $s$ is an invariant savings-labor productivities distribution. That is, the aggregate supply of savings equals the total capital. 
We define $H (s) =\int _{X}x_{1} s (d (x_{1} ,x_{2}))$. It is easy to see that $\succeq$ agrees with $ \succeq  _{SD}$ (see Section \ref{sec:unique}). 

The agents' utility from consumption is given by a utility function $u$ which is assumed to be strictly concave and strictly increasing. If the agents choose to save $a$ then their consumption in the current period is $R x_{1} +w x_{2} -a$. Thus, using the equilibrium conditions $R =f^{ \prime } (H(s)) - \delta +1$ and $w =f (H (s)) -f^{ \prime } (H (s)) H (s)$, in a Bewley-Aiyagari model the payoff function $\tilde{\pi} $ is given by
\begin{equation*}\tilde{\pi}  (x ,a ,H (s)) =u\left((f^{ \prime } (H (s)) - \delta + 1)x_{1} +(f (H (s)) -f^{ \prime } (H (s)) H (s)) x_{2} -a\right).
\end{equation*}
It is easy to establish that $G$ is single-valued and that Assumption \ref{Assumption 0} holds. Thus, the existence of an equilibrium in a Bewley-Aiyagari model follows from Theorem \ref{Theorem existence}.\footnote{Some of the previous existence results rely on the $X$-ergodicity of $Q$ (e,g., \cite{acikgoz2015existence}) or on monotonicity arguments (e.g., \cite{acemoglu2012}). The proof presented in this paper shows that these conditions are not needed in order to establish the existence of an equilibrium.}   

Under mild technical conditions on the utility function (for example, if $u$ is bounded or if $u$ belongs to the constant relative risk aversion class), the $X$-ergodicity of $Q$ can be proven in a similar manner to \cite{acikgoz2015existence}. It can be established also that the next period's
savings are increasing in the current period's savings, i.e., $\tilde{g}$ is increasing in $x_{1}$. Thus, to prove the uniqueness of an MFE in a Bewley-Aiyagari model, one needs to prove
that $\tilde{g}$ is decreasing in the aggregator $H(s)$. In a recent paper, \cite{Light2017} proves the uniqueness of an MFE for the special case that the agents' utility function is in the CRRA (constant relative risk aversion) class with a relative risk aversion coefficient
that is less than or equal to one, and the production function's elasticity of substitution is bounded below by $1$. Under these assumptions, we can use the results in \cite{Light2017} to show that $\tilde{g}$ is decreasing in the aggregator $H(s)$. Then, we can use Corollary \ref{Hetero Macro uniq} to prove the uniqueness of an MFE. As a note for future research, our results suggest that the result in \cite{Light2017}
could be generalized, weakening the conditions on the relative risk aversion and on the production function. With this, we believe our approach could be used to show uniqueness for a broader class of heterogeneous agent macro models. Finally, we note that the uniqueness result in \cite{hopenhayn1992entry}
can be obtained from Corollary \ref{Hetero Macro uniq} also. For the sake of brevity we omit the details here.

\section{Conclusions}
This paper studies the existence and uniqueness of an MFE in stochastic games with a general state space. We provide conditions that ensure
the uniqueness of an MFE. We also prove that there exists an MFE under continuity and concavity conditions on the primitives of the model. We show that
a general class of dynamic oligopoly models satisfies these conditions, and thus, these models have a unique MFE. Furthermore, we prove the existence of a unique MFE in heterogeneous agent macro models. We also derive general comparative statics results regarding the MFE and apply them to dynamic oligopoly models. 

We believe that our results can be applied to other models in operations research
and economics. For example, in order to analyze market design problems in online platforms, like in the reputation model we studied, it is natural to assume a large-scale MFE limit. Typical questions of interest in these contexts involve the market's response to platforms' market design choices. Hence, knowing that this response is unique and that one can predict its directional changes could significantly strengthen the analysis of these platforms.

We believe our results can be extended to prove the uniqueness of an invariant distribution for a general Markov chain where the next period's state  depends on the previous state and on the previous state's distribution. These Markov chains can capture other interesting applications in operations research, such as strategic queueing systems. We leave this direction for future work.

\bibliographystyle{ecta}
\bibliography{unique1}


\newpage

\appendix

\section{\label{Sec: Ext} Appendix: Extensions}

In this section we extend the model presented in Section \ref{Section: model}. In Section \ref{Sec: Coupling th} we study a model where the players are coupled through actions and in Section \ref{Sec: Ex-ante h} we study a model where the players are ex-ante heterogeneous. 

\subsection{Coupling Through Actions \label{Sec: Coupling th}} 

In this section we consider a model where the transition function and the payoff function of each player depend on both the states and the actions of all other players. The model is the same as the original model in Section \ref{Section: model} except that now the probability measure $s$ describes the joint distribution of players over actions and states and not only over states, that is, $s \in \mathcal{P}(X \times A)$. Thus, the transition function $w(x,a,s,\zeta)$ and the payoff function $\pi (x,a,s)$ depend on the joint distribution over state-action pairs $s \in \mathcal{P}(X \times A)$. We refer to $s \in \mathcal{P}(X \times A)$ as the population action-state profile and to the marginal distribution of the population action-state profile over $X$ as the population state (i.e., the population state's distribution is described by the  probability measure $s\left (\cdot ,A\right )$). 

An MFE is defined similarly to the definition in Section \ref{Section: model}. In an MFE, every player
conjectures that $s$ is the fixed long run population action-state profile, and plays according to a stationary strategy $g$. If every  player plays according to the strategy $g$ when the population action-state profile is $s$, then $s$ constitutes an invariant distribution. 

Given the stationary strategy $g$, $s \in \mathcal{P}(X \times A)$ is an invariant distribution if 
\begin{equation}s\left (B \times D\right ) =\int _{X}\int _{B}1_{D}(g\left (y ,s\right ))Q\left (x ,s ,dy\right )s\left (dx ,A\right ) =\int _{X}\overline{Q}(x ,s ,B \times D)s(dx ,A), \label{Equation: Action1}
\end{equation}
for all $B \times D \in \mathcal{B}(X \times A)$ where $Q(x,s,B) = \Pr (w(x ,s ,g(x ,s) ,\zeta ) \in B)$ and\protect\footnote{
Note that $\overline{Q}$ is a Markov kernel on $X \times A.$}    
\begin{equation*}\overline{Q}\left (x ,s ,B \times D\right ) =\int _{B}1_{D}(g\left (y ,s\right ))Q\left (x ,s ,dy\right ).
\end{equation*}
To see that Equation (\ref{Equation: Action1}) holds, first assume that $X$ and $A$ are discrete sets. The joint probability mass  function of a stationary distribution $s\left (y ,a\right )$ is given by \begin{equation*}s(y ,a) =s(y ,A)\overline{s}(a\vert y) =s(y ,A)1_{\{a\}}(g(y ,s))
\end{equation*} where $\overline{s}(a\vert y)$ is the probability of playing the action $a \in A$ given that the state is $y \in X$. Since the players use the pure strategy $g$ we have $\overline s(a\vert y)=1_{\{a\}}(g(y ,s))$.
Thus, \begin{equation*}s(B \times D) =\sum _{y \in B}\sum _{a \in D}s(y ,A)1_{\{a\}}(g(y ,s)) =\sum _{y \in B}s(y,A)1_{D}(g(y ,s)).
\end{equation*}

In addition, since $s$ is invariant, the marginal distribution $s\left ( \cdot  ,A\right )$ must satisfy $s(y ,A) =\sum _{x \in X}s(x ,A)Q(x ,s ,y)$. Thus,  

\begin{equation*}s(B \times D) =\sum _{x \in X}\sum _{y \in B}1_{D}\left (g\left (y ,s\right )\right)Q(x,s,y)s(x ,A).
\end{equation*}
Similarly, Equation (\ref{Equation: Action1}) holds in the general state space. 

If $A$ is compact then $X \times A$ is compact, and thus, $\mathcal{P}\left (X \times A\right )$ is compact in the weak topology.  Similar arguments to the arguments in the proof of Theorem \ref{Theorem existence} show that the operator $\overline{\Phi } :\mathcal{P} (X) \rightarrow \mathcal{P} (X)$ defined by
\begin{equation*}\overline{\Phi } s(B \times D) =\int  _{X}\overline{Q}(x,s,B \times D)s(dx,A).
\end{equation*}
is continuous (see more details in the proof of Theorem \ref{Theorem (Coupling actions)}). Thus,  as in the proof of Theorem \ref{Theorem existence}, we can apply Schauder-Tychonoff's fixed point theorem to prove that $\overline{\Phi}$ has a fixed point.

The uniqueness result holds under the same conditions as the conditions in Theorem \ref{Theorem uniq} except that the assumptions on the Markov kernel $Q$ in Assumption \ref{Assumption uniq} part (i) are assumed on the Markov kernel $\overline{Q}$. The proof of Theorem \ref{Theorem (Coupling actions)} part (i) is essentially the same as the proof of Theorem \ref{Theorem uniq}. Similarly, Theorem \ref{Theorem (Coupling actions)} part (iii) holds when the assumptions on the Markov kernel $Q$ are assumed on the Markov kernel $\overline{Q}$.  

We summarize the discussion in the following Theorem.

\begin{theorem}
\label{Theorem (Coupling actions)}Consider the model described in this section. Suppose that the action set $A$ is compact. 

(i) Under the assumptions of Theorem  \ref{Theorem uniq} where $Q$ is replaced by $\overline{Q}$  the MFE is unique.

(ii) Under the assumptions of Theorem \ref{Theorem existence} there exists an MFE. 

(iii) Let $(I,\succeq_{I})$ be a partially ordered set. Assume that $\overline{Q}$ is increasing in $e$ on $I$.Then, under the assumptions of part (i), the unique MFE $s (e)$ is increasing in the following sense: $e_{2}  \succeq _{I}e_{1}$ implies $s (e_{2}) \succeq s (e_{1})$.\protect\footnote{
Recall that we say that $\overline{Q}$ is increasing in $e$ if $\overline{Q} (x ,s ,e_{2} , \cdot ) \succeq  _{SD}\overline{Q} (x ,s ,e_{1}, \cdot )$ for all $x$, $s$, and all $e_{2} ,e_{1} \in I$ such that $e_{2}  \succeq _{I}e_{1}$. Note that the orders $\succeq _{SD}$ and $\succeq$ are on measures over state-action pairs.} 
\end{theorem}

The assumptions on $\overline{Q}$ that are needed in order to guarantee the uniqueness of an MFE can be verified in a similar manner to the assumptions on $Q$. In particular, in some models it is enough to show that the policy function $g(x,s)$ is increasing in the state $x$ and decreasing in the population action-state profile state $s$ which is a natural property in many dynamic oligopoly models (see Section \ref{Sec: DOM}). In Section \ref{sec:advertising} we prove that the policy  function $g(x,s)$ is increasing in $x$ and decreasing in $s$ in a dynamic advertising model where each player's payoff function depends on the other players' actions, and we use Theorem \ref{Theorem (Coupling actions)} to prove that the model has a unique MFE. 

\subsection{Ex-ante Heterogeneity \label{Sec: Ex-ante h}}
In this section we study a mean field model with ex-ante heterogeneous players. We assume that the players are heterogeneous in their payoff functions and in their transition functions. Assume that before the time horizon, each player has a type $\theta  \in \Theta$, where $\Theta $ is a finite partially ordered set. Each player's type is fixed throughout the horizon. Let $\Upsilon$ be the probability mass function over the type space; $\Upsilon (\theta )$ is the mass of players whose type is $\theta  \in \Theta$, which is common knowledge. Adding the argument $\theta  \in \Theta$ to the functions defined in Section \ref{Section: model}, we can modify the definitions of Section \ref{Section: model} to include the ex-ante heterogeneity of the players. In particular, we denote by $w (x ,a ,s ,\zeta  ,\theta )$ the transition function of type $\theta  \in \Theta $ and by $\pi  (x ,a ,s ,\theta )$ the payoff function of type $\theta  \in \Theta $.

Let $X_{h} =X \times \Theta$ be an extended state space for the mean field model with ex-ante heterogeneous players. If a player's extended state is $x_{h} =(x ,\theta ) \in X_{h}$ then the player's state is $x$ and the player's type is $\theta $. Let $s_{h}$ be the population state over the extended state space, i.e., $s_{h} \in \mathcal{P}(X \times \Theta )$. 

For a probability measure $s_{h} \in \mathcal{P}(X \times \Theta )$, define a probability measure $S\left (s_{h}\right ) \in \mathcal{P}(X)$ by
\begin{equation*}S\left (s_{h}\right )(B) =\sum _{\theta  \in \Theta }s_{h}(B ,\theta )
\end{equation*} for all $B \in \mathcal{B}(X)$. That is, $S(s_{h})$ is the marginal distribution of $s_{h}$ that describes the population state. 

For the model with ex-ante heterogeneous players we define the payoff function $\pi _{h}\left (x_{h} ,a ,s_{h}\right ) =\pi  (x ,a ,S\left (s_{h}\right ) ,\theta )$. Note that we consider a model where each player's payoff function depends on the other players' states (the population state) and not on the other players' types. This seems reasonable in most applications, as types usually represent ex-ante heterogeneity in the payoff functions, discount factors, etc. We now define the transition function.

For a fixed extended population state $s_{h} \in \mathcal{P}(X \times \Theta )$ and a strategy $g(x ,S\left (s_{h}\right ) ,\theta )$, the probability that player $i$'s next period's state will lie in a set $B \times D \in \mathcal{B}(X) \times 2^{\Theta }$, given that her current state is $x_{h} =(x,\theta) \in X_{h}$, her type is $\theta $, and she takes the action $a =g (x ,S\left (s_{h}\right ) ,\theta )$, is:
\begin{equation*}Q_{h} (x_{h} ,s_{h} ,B \times D) =\Pr  (w(x,g(x,S\left (s_{h}\right ) ,\theta ) ,S\left (s_{h}\right ) ,\zeta  ,\theta ) \in B)1_{D}\left (\theta \right).
\end{equation*}

These definitions map the payoff function and transition function in the model with ex-ante heterogeneous players to the model with ex-ante homogeneous players that we considered in Section \ref{Section: model}. Thus, all the results in this paper hold also in the case of ex-ante heterogeneity where the assumptions that we made on $\pi $, $w$ and $Q$ are now assumed on $\pi _{h}$, $w_{h}$ and $Q_{h}$. Thus, all our results can easily be extended to the case of ex-ante heterogeneous players. Note that in this model, players of different types may play different MFE strategies. We now provide more details. 

Similarly to Section \ref{Section: model}, in an MFE every player plays according to the strategy $g$ when the extended population state is $s_{h}$ and $s_{h}$ constitutes an invariant distribution given the strategy $g$. That is, $s_{h}$ satisfies 
\begin{equation*}
    s_{h}(B \times D) = \int _{X_{h}} Q_{h} (x_{h} ,s_{h} ,B \times D) s_{h}(dx_{h})
\end{equation*}
for all $B \times D \in \mathcal{B}(X) \times 2^{\Theta}$. 

The following theorem follows immediately from the results in the main text when $Q$ is replaced by $Q_{h}$. Note that $X_{h} = X \times \Theta$ is a product space so we can use Theorem \ref{Corr:product space} instead of Theorem \ref{Theorem uniq} to prove the uniqueness of an MFE.  

\begin{theorem} \label{Thm: ex-ante het}
\label{Theorem: Ex-ante het} Consider the model described in this section. 

(i) Under the assumptions of Theorem \ref{Corr:product space} (with the state space $X\times \Theta$) where $Q$ is replaced by $Q_{h}$, the MFE is unique.

(ii) Under the assumptions of Theorem \ref{Theorem existence} there exists an MFE. 

(iii) Let $(I,\succeq_{I})$ be a partially ordered set. Assume that $Q_{h}$ is increasing in $e$ on $I$. Then, under the assumptions of part (i), the unique MFE $s_{h} (e)$ is increasing in the following sense: $e_{2}  \succeq _{I}e_{1}$ implies $s _{h} (e_{2}) \succeq s_{h} (e_{1})$.
\end{theorem}

We define the $X$-transition function of a type $\theta$ player by 
$$Q_{\theta}(x,s_{h},B) = \Pr  (w(x,g(x,S\left (s_{h}\right ) ,\theta ) ,S\left (s_{h}\right ) ,\zeta  ,\theta ) \in B) $$
for all $B \in \mathcal{B}(X)$.
As discussed in Section \ref{sec:unique}, the key assumption that implies the uniqueness of an MFE is related to the transition function's monotonicity properties. In particular, the assumption is that the transition function is increasing in the players' own states and decreasing in the extended population state. In the case of ex-ante heterogeneity, the next Lemma shows that if the transition function of each player $Q_{\theta}$ is increasing in $x$ and decreasing in $s_{h}$ for every type $\theta$ then $Q_{h}$ is increasing in $x$ and decreasing in $s_{h}$ with respect to $x$. This fact is useful for applications when we want to verify the monotonicity conditions needed in Theorem \ref{Thm: ex-ante het} part (i) that imply the uniqueness of an MFE. 

\begin{lemma}\label{Lemma: Ex-ante}
Assume that $Q_{\theta}$ is increasing in $x$ and decreasing in $s_{h}$ for every type $\theta$. Then $Q_{h}$ is increasing in $x$ and decreasing in $s_{h}$ with respect to $x$. 
\end{lemma}

\section{Appendix: Proofs}

\subsection{Uniqueness: Proof of Theorem \ref{Corr:product space}} 
\begin{proof} [Proof of Theorem~\ref{Corr:product space}]
Assume without loss of generality that $Q$ is increasing in $x_{1}$ and decreasing in $s$ with respect to $x_{1}$. 
 
For $s_{1} ,s_{2} \in \mathcal{P} (X)$ we write $s_{1}  \succeq _{SD ,X_{1}}s_{2}$ if for all functions $f :X_{1} \times X_{2} \rightarrow \mathbb{R}$ that are increasing in the first argument (i.e., $x_{1}^{ \prime } \geq x_{1}$ implies that $f(x_{1}^{ \prime } ,x_{2}) \geq f(x_{1} ,x_{2})$ for all $x_{2} \in X$) we have
\begin{equation*}\int _{X}f (x_{1} ,x_{2}) s_{1} (d (x_{1} ,x_{2})) \geq \int _{X}f (x_{1} ,x_{2}) s_{2} (d (x_{1} ,x_{2})).
\end{equation*}
We note that
if $ \succeq $ agrees with $ \succeq  _{SD}$, then $ \succeq $ agrees with $ \succeq _{SD ,X_{1}}$ (recall that $s_{2}  \succeq  _{SD}s_{1}$ if the last inequality holds for every increasing function $f :X_{1} \times X_{2} \rightarrow \mathbb{R}$). 

Let $f :X_{1} \times X_{2} \rightarrow \mathbb{R}$ be increasing in the first argument, $\theta _{1} ,\theta _{2} \in \mathcal{P} (X)$ and assume that $\theta _{1}  \succeq _{SD ,X_{1}}\theta _{2}$. Let $s_{1} ,s_{2}$ be two MFEs such that $s_{2} \succeq s_{1}$. We have
\begin{align*}\int _{X}f (y_{1} ,y_{2}) M_{s_{2}} \theta _{2} (d (y_{1} ,y_{2})) &  =\int _{X}\int _{X}f(y_{1} ,y_{2})Q((x_{1} ,x_{2}) ,s_{2} ,d(y_{1} ,y_{2}))\theta _{2}(d (x_{1} ,x_{2}))  \\
 & \leq \int _{X}\int _{X}f(y_{1} ,y_{2})Q((x_{1} ,x_{2}) ,s_{2} ,d(y_{1} ,y_{2}))\theta _{1}(d (x_{1} ,x_{2}))  \\
 & \leq \int _{X}\int _{X}f(y_{1} ,y_{2})Q((x_{1} ,x_{2}) ,s_{1} ,d(y_{1} ,y_{2}))\theta _{1}(d (x_{1} ,x_{2}))  \\
 &  =\int f (x_{1} ,x_{2}) M_{s_{1}} \theta _{1} (d (x_{1} ,x_{2})).\end{align*}
 Thus, $M_{s_{1}} \theta _{1}  \succeq _{SD ,X_{1}}M_{s_{2}} \theta _{2}$. The first inequality follows from the facts that $f$ is increasing in the first argument, $Q$ is increasing in $x_{1}$, and $\theta _{1}  \succeq _{SD ,X_{1}}\theta _{2}$.  The second inequality follows from the fact that $Q$ is decreasing in $s$ with respect to $x_{1}$.

We conclude that $M_{s_{1}}^{n} \theta _{1}  \succeq _{SD ,X_{1}}M_{s_{2}}^{n} \theta _{2}$ for all $n \in \mathbb{N}$. $Q$ being $X$-ergodic implies that $M_{s_{i}}^{n} \theta _{i}$ converges weakly to $\mu _{s_{i}} =s_{i}$. Since $ \succeq _{SD ,X_{1}}$ is a closed order, we have $s_{1}  \succeq _{SD ,X_{1}}s_{2}$ which implies that $s_{1} \succeq s_{2}$. The rest of the proof is the same as the proof of Theorem \ref{Theorem uniq}. 
\end{proof}

\subsection{Existence: Proofs of Theorem \ref{Theorem existence} and Lemma \ref{Lemma concave}}
We first introduce preliminary notation and results. 

Let $B (X \times \mathcal{P} (X))$ be the space of all bounded functions on $X \times \mathcal{P} (X)$. Define the operator $T :B (X \times \mathcal{P} (X)) \rightarrow B (X \times \mathcal{P} (X))$ by
\begin{equation*}Tf (x ,s) =\underset{a \in \Gamma  (x)}{\max }h (x ,a ,s ,f)
\end{equation*}where
\begin{equation*}h (x ,a ,s ,f) =\pi  (x ,a ,s) +\beta  \sum _{j =1}^{n}p_{j} f (w (x ,a ,s ,\zeta _{j}) ,s).
\end{equation*}
The operator $T$ is called the Bellman operator.

\begin{lemma}
\label{Lemma 1}The optimal strategy correspondence $G (x ,s)$ is non-empty, compact-valued and upper hemicontinuous. 
\end{lemma}

\begin{proof}
Assume that $f \in B (X \times \mathcal{P} (X))$ is (jointly) continuous. Then for each $\zeta  \in E$, $f (w (x ,a ,s ,\zeta ) ,s)$ is continuous as the composition of continuous functions. Thus, $h (x ,a ,s ,f)$ is continuous as the sum of continuous functions. Since $\Gamma  (x)$ is continuous, the maximum theorem (see Theorem 17.31 in \cite{aliprantis2006infinite})
implies that $T f (x ,s)$ is jointly continuous. 

We conclude that for all $n =1 ,2 ,3 \ldots$, $T^{n} f$ is continuous. Under Assumption \ref{Assumption 0}, standard dynamic programming arguments (see \cite{bertsekas1978stochastic}) show that
$T^{n} f$ converges to $V$ uniformly. Since the set of continuous functions is closed under uniform convergence, $V$ is continuous. Thus, $h(x,a,s,V)$ is continuous. From the maximum theorem, $G(x,s)$ is non-empty, compact-valued and upper hemicontinuous. 
\end{proof}

We say that $k_{n} :X \rightarrow \mathbb{R}$ converges continuously to $k$ if $k_{n} (x_{n}) \rightarrow k (x)$ whenever $x_{n} \rightarrow x$. The following Proposition is a special case of Theorem 3.3 in \cite{serfozo1982convergence}.

\begin{proposition}
\label{Prop1}Assume that $k_{n} :X \rightarrow \mathbb{R}$ is a uniformly bounded sequence of functions. If $k_{n} :X \rightarrow \mathbb{R}$ converges continuously to $k$ and $s_{n}$ converges weakly to $s$ then
\begin{equation*}\underset{n \rightarrow \infty }{\lim }\int _{X}k_{n} (x) s_{n} (d x) =\int _{X}k (x) s (d x).
\end{equation*}
\end{proposition}

In order to establish the existence of an MFE, we will use the following Proposition (see Corollary 17.56 in \cite{aliprantis2006infinite}).

\begin{proposition}
(Schauder-Tychonoff) Let $K$ be a nonempty, compact, convex subset of a locally convex Hausdorff space, and let $f :K \rightarrow K$ be a continuous function. Then the set of fixed points of $f$ is compact and nonempty. 
\end{proposition}

\begin{proof}
[Proof of Theorem~\ref{Theorem existence}] Let $g (x ,s) =G (x ,s)$ be the unique optimal stationary strategy. From Lemma \ref{Lemma 1},  $g$ is continuous. 

Consider the operator $\Phi  :\mathcal{P} (X) \rightarrow \mathcal{P} (X)$ defined by
\begin{equation*}\Phi  s (B) =\int _{X}Q_{g} (x ,s ,B) s (d x).
\end{equation*}
If $s$ is a fixed point of $\Phi $ then $s$ is an MFE. Since $X$ is compact $\mathcal{P} (X)$ is compact (i.e., compact in the weak topology, see \cite{aliprantis2006infinite}).
Clearly $\mathcal{P} (X)$ is convex. $\mathcal{P} (X)$ endowed with the weak topology is a locally convex Hausdorff space. If $\Phi $ is continuous, we can apply Schauder-Tychonoff's fixed point theorem to prove that $\Phi $ has a fixed point. We now show that $\Phi $ is continuous. 

First, note that for every bounded and measurable function $f :X \rightarrow \mathbb{R}$ and for every $s \in \mathcal{P} (X)$ we have
\begin{equation}\int _{X}f (x) \Phi  s (d x) =\int _{X}\sum _{j}p_{j} f (w (x ,g (x ,s) ,s ,\zeta _{j}) s (d x). \label{1.1}
\end{equation}To see this, first assume that $f =1_{B}$ where $1_{B}$ is the indicator function of $B \in \mathcal{B} (X)$. Then
\begin{align*}\int _{X}f (x) \Phi  s (d x) &  =\int _{X}1_{B} \Phi  s (d x) \\
 &  =\int _{X}Q_{g} (x ,s ,B) s (d x) \\
 &  =\int _{X}\sum _{j}p_{j} 1_{B} (w (x ,g (x ,s) ,s ,\zeta _{j}) s (d x) \\
 &  =\int _{X}\sum _{j}p_{j} f (w (x ,g (x ,s) ,s ,\zeta _{j}) s (d x). \end{align*}A standard argument shows that (\ref{1.1}) holds for every bounded and measurable function $f$. 

Assume that $s_{n}$ converges weakly to $s$. Let $f :X \rightarrow \mathbb{R}$ be a continuous and bounded function. Since $w$ is jointly continuous and $g$ is continuous(see Lemma \ref{Lemma 1}), we have $$f (w (x_{n} ,g (x_{n} ,s_{n}) ,s_{n} ,\zeta )) \rightarrow f(w (x ,g (x ,s) ,s ,\zeta )$$ whenever $x_{n} \rightarrow x$. 
Let $k_{n} (x):=\sum _{j =1}^{n}p_{j} f (w (x ,g (x ,s_{n}) ,s_{n} ,\zeta _{j})$ and 

$k (x):=\sum _{j =1}^{n}p_{j} f (w (x ,g (x ,s) ,s ,\zeta _{j})$. Then $k_{n}$ converges continuously to $k$, i.e., $k_{n} (x_{n}) \rightarrow k (x)$ whenever $x_{n} \rightarrow x$. Since $f$ is bounded, the sequence $k_{n}$ is uniformly bounded. Using Proposition \ref{Prop1} and equality (\ref{1.1}),
we have
\begin{align*}\underset{n \rightarrow \infty }{\lim }\int _{X}f (x) \Phi  s_{n} (d x) &  =\underset{n \rightarrow \infty }{\lim }\int _{X}k_{n} (x) s_{n} (d x) \\
 &  =\int _{X}k (x) s (d x) \\
 &  =\int _{X}f (x) \Phi  s (d x)\;\;\text{.}\;\;\end{align*}Thus, $\Phi  s_{n}$ converges weakly to $\Phi  s$. We conclude that $\Phi $ is continuous. Thus, by the Schauder-Tychonoff's fixed point theorem, $\Phi $ has a fixed point. 
\end{proof}

\vspace{5mm}

\begin{proof}
[Proof of Lemma~\ref{Lemma concave}] Assume that $f \in B (X \times \mathcal{P} (X))$ is concave and increasing in $x$. Since the composition of a concave and increasing function with a concave function is a concave function, the function $f (w (x ,a ,s ,\zeta ) ,s)$ is concave in $(x ,a)$ for all $s$ and $\zeta $. Since $w$ and $f$ are increasing in $x$ then $f (w (x ,a ,s ,\zeta ) ,s)$ is increasing in $x$ for all $a$, $s$ and $\zeta $. Thus, $h (x ,a ,s ,f)$ is concave in $(x ,a)\;$and increasing in $x$ as the sum of concave and increasing functions. A standard argument shows that $T f$ is increasing in $x$. Proposition 2.3.6 in \cite{bertsekas2003convex}
and the fact that $\Gamma  (x)$ is convex-valued imply that $T f (x ,s) =\underset{a \in \Gamma  (x)}{\max }h (x ,a ,s ,f)$ is concave in $x$. 

We conclude that for all $n =1 ,2 ,3 \ldots$, $T^{n} f$ is concave and increasing in $x$. Standard dynamic programming arguments (see \cite{bertsekas1978stochastic}) show that $T^{n} f$ converges to $V$ uniformly. Since the set of concave and increasing functions is closed under uniform convergence, $V$ is concave and increasing in $x$. 

Since $\pi $ is strictly concave in $a$, $h (x ,a ,s ,V)$ is strictly concave in $a$. This implies that $G(x ,s) =\ensuremath{\operatorname*{argmax}}_{a \in \Gamma (x)}h(x ,a ,s ,V)$ is single-valued which proves the Lemma. 
\end{proof}

\subsection{Comparative statics: Proof of Theorem \ref{Theorem comp}}

\begin{proof}
[Proof of Theorem~\ref{Theorem comp}] 
Under the assumptions of Theorem \ref{Theorem uniq},
the operator $M_{s} :\mathcal{P} (X) \times I \rightarrow \mathcal{P} (X)$ defined by
\begin{equation*}M_{s} (\theta  ,e) ( \cdot ) =\int _{X}Q (x,s,e, \cdot)\theta (d x)
\end{equation*}has a unique fixed point $\mu _{s ,e}$ for each $s \in \mathcal{P} (X)$ and $e \in I$.  

Fix $s \in \mathcal{P} (X)$. Let $\theta _{2}  \succeq  _{SD}\theta _{1}$ and $e_{2}  \succeq _{I}e_{1}$ and let $B$ be an upper set. We have
\begin{align*}M_{s} (\theta _{2} ,e_{2}) (B) &  =\int _{X}Q (x ,s ,e_{2} ,B) \theta _{2} (d x) \\
 &  \geq \int _{X}Q (x ,s ,e_{2} ,B) \theta _{1} (d x) \\
 &  \geq \int _{X}Q (x ,s ,e_{1} ,B) \theta _{1} (d x) =M_{s} (\theta _{1} ,e_{1}) (B).\end{align*}Thus, $M_{s} (\theta _{2} ,e_{2})  \succeq  _{SD}M_{s} (\theta _{1} ,e_{1})$. The first inequality holds because $\theta _{2}  \succeq  _{SD}\theta _{1}$ and $Q$ is increasing in $x$ when $B$ is an upper set. The second inequality follows from the fact that $Q\;$is increasing in $e$ when $B$ is an upper set. 

We conclude that $M_{s}$ is an increasing function from $\mathcal{P} (X) \times I$ into $\mathcal{P} (X)$ when $\mathcal{P} (X)$ is endowed with $ \succeq  _{SD}$. Thus,  $M_{s}^{n}(\theta _{2} ,e_{2})  \succeq  _{SD}M_{s}^{n} (\theta _{1} ,e_{1})$ for all $n \in \mathbb{N}$. $Q$ being $X$-ergodic implies that $M_{s}^{n}(\theta _{i} ,e_{i}) $ converges weakly to $\mu _{s ,e_{i}}$. Since $ \succeq  _{SD}$ is closed under weak convergence (see \cite{kamae1977stochastic}), we have  $\mu _{s ,e_{2}}  \succeq  _{SD}\mu _{s ,e_{1}}$.  

Now assume that $e_{2}  \succeq _{I}e_{1}$ and let $s (e_{2}) ,s (e_{1})$ be the corresponding MFEs. Assume in contradiction that $s (e_{2}) \prec s (e_{1})$. From the same argument as in Theorem \ref{Theorem uniq} we can conclude that $\mu _{s (e_{2}) ,e}  \succeq  _{SD}\mu _{s (e_{1}) ,e}$ for each $e \in I$. Note that $s (e)$ is an MFE if and only if $s (e) =\mu _{s (e) ,e}$. We have
\begin{equation*}s (e_{2}) =\mu _{s (e_{2}) ,e_{2}}  \succeq  _{SD}\mu _{s (e_{2}) ,e_{1}}  \succeq  _{SD}\mu _{s (e_{1}) ,e_{1}} =s (e_{1}).
\end{equation*}Transitivity of $ \succeq  _{SD}$ implies $s (e_{2})  \succeq  _{SD}s (e_{1})$. But since $ \succeq  _{SD}$ agrees with $ \succeq $, $s (e_{2})  \succeq  _{SD}s (e_{1})$ implies $s (e_{2}) \succeq s (e_{1})$ which is a contradiction. We conclude that $s (e_{2}) \succeq s (e_{1})$. 
\end{proof}

\subsection{Dynamic Oligopoly Models: Proofs of Theorems \ref{Unique DOP}, \ref{Theorem DOP comp}, \ref{Theorem:dynamic advertising}, and \ref{THM: REPUTATION}}
\begin{proof}
[Proof of Theorem~\ref{Unique DOP}] The idea of the proof is to show that the conditions of Theorem \ref{Theorem uniq} and Theorem \ref{Theorem existence} hold. In Lemma \ref{Lemma single-valued}
we prove that the optimal stationary investment strategy is single-valued. In Lemma \ref{Lemma Q is increasing}
we prove that $Q$ is increasing in $x$ and decreasing in $s$. In Lemma \ref{Lemma compact state} we prove that the state space can be chosen to
be compact. That is, there exists a compact set $\bar{X} =[0 ,\bar{x}]$ such that $Q (x ,s ,\bar{X}) =1$ whenever $x \in \bar{X}$ and all $s \in \mathcal{P} (X)$. This means that if a firm's initial state is in $\bar{X}$, then the firm's state will remain in $\bar{X}$ in the next period with probability $1$. In Lemma \ref{Lemma Q is ergodic} we prove that $Q$ is $\bar{X}$-ergodic. Thus, all conditions from Theorem \ref{Theorem uniq} and Theorem \ref{Theorem existence} hold and we conclude that the model has a unique MFE. 
\end{proof}

We first introduce some notations. 

Let $B (X \times \mathcal{P} (X))$ be the space of all bounded functions on $X \times \mathcal{P} (X)$. For $f \in B (X \times \mathcal{P} (X))$ define
\begin{equation*}f_{x} (x ,s):=\frac{ \partial f (x ,s)}{ \partial x}.
\end{equation*}For the rest
of the paper we say that $f \in B (X \times \mathcal{P} (X))$ is differentiable if it is differentiable in the first argument. Similarly, we write $u_{x} (x ,s)$ to denote the derivative of $u$ with respect to $x$. 

For the proof of the theorem, it will be convenient to change the decision variable in the Bellman equation.
Define $$z =(1 -\delta ) x +k (a),$$ and note that we can write $a =k^{ -1} (z -(1 -\delta ) x)$, which is well defined because $k$ is strictly increasing. The resulting Bellman operator is given by
\begin{equation*}Kf(x ,s) =\underset{z \in \Gamma  (x)}{\max }J (x ,z ,s ,f),
\end{equation*}where $\Gamma  (x) =[(1 -\delta )x +k (0) ,(1 -\delta ) x +k (\bar{a})]$ and
\begin{equation*}J (x ,z ,s ,f) =\pi(x,z,s) +\beta  \sum _{j}p_{j} f (z \zeta _{j} ,s),
\end{equation*}
where $\pi(x,z,s)=u (x ,s) -d k^{ -1} (z -(1 -\delta ) x)$.

Let $\mu _{f} (x ,s) =\operatorname{argmax}_{z \in \Gamma  (x)}J (x ,z ,s ,f)$ and $\mu  (x ,s) =\operatorname{argmax}_{z \in \Gamma  (x)}J (x ,z ,s ,V)$. Note that $\mu  (x ,s) =(1 -\delta ) x +k (g (x ,s))$ where $g$ is the optimal stationary investment strategy. With this change of variables, we can use the \emph{envelope}
theorem (see \cite{benveniste1979}). Since $u$ and $k$ are continuously differentiable, then $J (x ,z ,s ,f)$ is continuously differentiable in $x$. The envelope theorem implies that $K f$ is differentiable and
\begin{equation*}K f_{x} (x ,s) =\frac{ \partial \pi  (x ,\mu _{f} (x ,s) ,s)}{ \partial x} =u_{x} (x ,s) +d (1 -\delta ) (k^{ -1})^{ \prime } (\mu _{f} (x ,s) -(1 -\delta ) x).
\end{equation*}

\begin{lemma} \label{lemma:mu}
\label{Lemma single-valued} $\mu  (x ,s)$ is single-valued, increasing in $x$ and decreasing in $s$. 
\end{lemma}

\begin{proof}
The main step of the proof is to show that if $f \in B (X \times \mathcal{P} (X))$ has decreasing differences then $Kf \in B (X \times \mathcal{P} (X))$ has decreasing differences. This implies that the value function $V$ has decreasing differences. An application of a Theorem by Topkis implies that $\mu  (x ,s)$ is increasing in $x$ and decreasing in $s$. Single-valuedness of $\mu$ follows from the concavity of the value function. We provide the details below. 

Assume that $f \in B (X \times \mathcal{P} (X))$ is concave in $x$, differentiable, and has decreasing differences. The function $f (z \zeta  ,s)$ is concave and increasing in $z$ for all $s$ and $\zeta $. Since $k$ is strictly concave and strictly increasing, $k^{ -1}$ is strictly convex and strictly increasing. This implies that $ -k^{ -1} (z -(1 -\delta ) x)$ is concave in $(x ,z)$. Thus, $J (x ,z ,s ,f)$ is concave in $(x ,z)$ as the sum of concave functions. Proposition 2.3.6 in \cite{bertsekas2003convex}
and the fact that $\Gamma  (x)$ is convex-valued imply that $K f (x ,s)$ is concave in $x$. 

Since $f$ has decreasing differences, then $f (z \zeta  ,s)$ has decreasing differences in $(z ,s)$ for all $\zeta $. Thus, $J$ has decreasing differences in $(z ,s)$ as the sum of functions with decreasing differences. From
Theorem 6.1 in \cite{topkis1978}, $\mu _{f} (x ,s)$ is decreasing in $s$ for every $x$. 

Let $x_{2} \geq x_{1}$, $z_{2} \geq z_{1}$, $y^{ \prime } =z_{1} -(1 -\delta ) x_{2}$, $y =z_{1} -(1 -\delta ) x_{1}$ and $t =z_{2} -z_{1}$. Note that $y \geq y^{ \prime }$. Convexity of $k^{ -1}$ implies that for $y \geq y^{ \prime }$ and $t \geq 0$, we have $k^{ -1} (y +t) -k^{ -1} (y) \geq k^{ -1} (y^{ \prime } +t) -k^{ -1} (y^{ \prime })$. That is, $k^{ -1} (z -(1 -\delta ) x)$ has decreasing differences in $(x ,z)$. Thus, $\pi  (x ,z ,s) =u (x ,s) -k^{ -1} (z -(1 -\delta ) x)$ has increasing differences in $(x ,z)$. 

Let $s_{2} \succeq s_{1}$. For every $x \in X$ we have
\begin{align}K f_{x} (x ,s_{1}) &  =\pi _{x} (x ,\mu _{f} (x ,s_{1}) ,s_{1}) \notag \\
 &  \geq \pi _{x} (x ,\mu _{f} (x ,s_{1}) ,s_{2}) \notag \\
 &  \geq \pi _{x} (x ,\mu _{f} (x ,s_{2}) ,s_{2}) =K f_{x} (x ,s_{2}). \label{eq:Kfx}\end{align}The first and last equality follow from the envelope theorem. The first inequality follows since $\pi $ has decreasing differences in $(x ,s)$. The second inequality follows from the facts that $\pi $ has increasing differences in $(x ,z)$ and $\mu _{f} (x ,s_{1}) \geq \mu _{f} (x ,s_{2})$. Thus, $K f$ has decreasing differences. 

Define $f^{n} =K^{n} f:=K (K^{n -1} f)$ for $n =1 ,2 ,\ldots $ where $K^{0} f:=f$. By iterating the previous argument we conclude that $f_{x}^{n} (x ,s)$ is decreasing in $s$ and $f^{n} (x ,s)$ is concave in $x$ for every $n \in \mathbb{N}$. 

Standard dynamic programming arguments (see \cite{bertsekas1978stochastic}) show that  $f^{n}$ converges uniformly to $V$. Since the set of concave functions is closed under uniform convergence, $V$ is concave in $x$. The envelope theorem implies that $f_{x}^{n} (x ,s) =\pi _{x} (x ,\mu _{f^{n}} (x ,s) ,s)$ for every $n \in \mathbb{N}$. Since $J (x ,z ,s ,f^{n})$ is strictly concave in $z$ when $f^{n}$ is concave, $\mu _{f^{n}}$ is single-valued. Theorem 3.8 and Theorem 9.9 in \cite{stokey1989}
show that $\mu _{f^{n}} \rightarrow \mu $. Thus, $f_{x}^{n} (x ,s) =\pi _{x} (x ,\mu _{f^{n}} (x ,s) ,s) \rightarrow \pi _{x} (x ,\mu  (x ,s) ,s) =V_{x} (x ,s)$. Using \eqref{eq:Kfx}, we conclude that $V_{x} (x ,s)$ is decreasing in $s$; hence, $V$ has decreasing differences. The same argument as above shows that $J (x ,z ,s ,V)$ has decreasing differences in $(z ,s)$ and increasing differences in $(x ,z)$. Since $J (x ,z ,s ,V)$ is strictly concave in $z$, then $\mu $ is single-valued. It is easy to see that $\Gamma  (x)$ is ascending in the sense of \cite{topkis1978} (i.e., for $x_{2} \geq x_{1}$ if $z \in \Gamma  (x_{2})$ and $z^{ \prime } \in \Gamma  (x_{1})$ then $\max \{z ,z^{ \prime }\} \in \Gamma  (x_{2})$ and $\min \{z ,z^{ \prime }) \in \Gamma  (x_{1})$). Theorem 6.1 in \cite{topkis1978} implies that $\mu  (x ,s)$ is increasing in $x$ and decreasing in $s$ which proves the Lemma. 
\end{proof}

\begin{lemma}
\label{Lemma Q is increasing}$Q$ is increasing in $x$ for each $s \in S$ and decreasing in $s$ for each $x \in X$. 
\end{lemma}

\begin{proof}
For each $s \in \mathcal{P} (X)$, $x_{2} \geq x_{1}$ and any upper set $B$ we have
\begin{align*}Q (x_{2} ,s ,B) &  =\Pr  (((1 -\delta ) x_{2} +k (g (x_{2} ,s)) \zeta ) \in B) \\
 &  =\Pr  (\mu  (x_{2} ,s) \zeta  \in B) \\
 &  \geq \Pr  (\mu  (x_{1} ,s) \zeta  \in B) =Q (x_{1} ,s ,B),\end{align*}where the inequality follows since $\mu $ is increasing in $x$. Thus, $Q (x_{2} ,s , \cdot ) \succeq  _{SD}Q (x_{1} ,s , \cdot )$. 

Similarly since $\mu  (x ,s)$ is decreasing in $s$, $Q$ is decreasing in $s$ for each $x \in X$. 
\end{proof}

We prove the following useful auxiliary lemma.

\begin{lemma}
\label{Lemma properties}(i) $\mu  (x ,s)$ is strictly increasing in $x$. 

(ii) For all $s \in \mathcal{P} (X)$, $\mu $ is Lipschitz-continuous in the first argument with a Lipschitz constant $1$. That is,
\begin{equation*}\vert \mu  (x_{2} ,s) -\mu  (x_{1} ,s)\vert  \leq \vert x_{2} -x_{1}\vert,
\end{equation*}for all $x_{2} ,x_{1}$ and $s \in \mathcal{P} (X)$. 
\end{lemma}

\begin{proof}
(i) Fix $s \in \mathcal{P} (X)$. Assume in contradiction that $x_{2} >x_{1}$ and $\mu  (x_{1} ,s) =\mu  (x_{2} ,s)$. First note that $\mu  (x_{1} ,s) =\mu  (x_{2} ,s) \geq (1 -\delta ) x_{2} +k (0) >(1 -\delta ) x_{1} +k (0):=\min \Gamma  (x_{1})$. Thus, $\min \Gamma  (x_{1}) <\mu  (x_{1} ,s) \leq \max \Gamma  (x_{1}) <\max \Gamma  (x_{2})$. We have
\begin{align*}0 &  \leq  -d (k^{ -1})^{ \prime } (\mu  (x_{1} ,s) -(1 -\delta ) x_{1}) +\beta  \sum _{j =1}^{n}p_{j} \zeta _{j} V_{x} (\mu  (x_{1} ,s) \zeta _{j} ,s) \\
 &  < -d (k^{ -1})^{ \prime } (\mu  (x_{2} ,s) -(1 -\delta ) x_{2}) +\beta  \sum _{j =1}^{n}p_{j} \zeta _{j} V_{x} (\mu  (x_{2} ,s) \zeta _{j} ,s),\end{align*}which contradicts the optimality of $\mu  (x_{2} ,s)$, since $\mu  (x_{2} ,s) <\max \Gamma  (x_{2})$. The first inequality follows from the first order condition (recall that $\min \Gamma  (x_{1}) <\mu  (x_{1} ,s)$). The second inequality follows from the fact that $k^{ -1}$ is strictly convex, which implies that $(k^{ -1})^{ \prime }$ is strictly increasing. Thus, $\mu $ is strictly increasing in $x$.

(ii) Fix $s \in \mathcal{P} (X)$. Let $x_{2} >x_{1}$. If $\mu  (x_{1} ,s) =\max \Gamma  (x_{1}) =(1 -\delta ) x_{1} +k (\bar{a})$, then
\begin{equation*}\mu  (x_{2} ,s) -\mu  (x_{1} ,s) \leq (1 -\delta ) (x_{2} -x_{1}) +k (\bar{a}) -k (\bar{a}) \leq x_{2} -x_{1}.
\end{equation*}So we can assume
that $\mu  (x_{1} ,s) <\max \Gamma  (x_{1})$. Assume in contradiction that $\mu  (x_{2} ,s) -\mu  (x_{1} ,s) >x_{2} -x_{1}$. Then $\mu  (x_{2} ,s) -(1 -\delta ) x_{2} >\mu  (x_{1} ,s) -(1 -\delta ) x_{1}$. We have
\begin{align*}0 &  \geq  -d (k^{ -1})^{ \prime } (\mu  (x_{1} ,s) -(1 -\delta ) x_{1}) +\beta  \sum _{j =1}^{n}p_{j} \zeta _{j} V_{x} (\mu  (x_{1} ,s) \zeta _{j} ,s) \\
 &  > -d (k^{ -1})^{ \prime } (\mu  (x_{2} ,s) -(1 -\delta ) x_{2}) +\beta  \sum _{j =1}^{n}p_{j} \zeta _{j} V_{x} (\mu  (x_{2} ,s) \zeta _{j} ,s).\end{align*}The first inequality follows from the first order condition. The second inequality follows from the facts that $(k^{ -1})$ is strictly convex and $V$ is concave (see the proof of Lemma \ref{lemma:mu}). The last inequality implies that $\mu  (x_{2} ,s) =\min \Gamma  (x_{2}) =(1 -\delta ) x_{2} +k (0)$. But $\mu  (x_{1} ,s) \geq \min \Gamma  (x_{1})$ implies
\begin{equation*}\mu  (x_{2} ,s) -\mu  (x_{1} ,s) \leq (1 -\delta ) (x_{2} -x_{1}) <x_{2} -x_{1},
\end{equation*}which is a contradiction.
We conclude that $\mu $ is Lipschitz-continuous in the first argument with a Lipschitz constant $1$. 
\end{proof}

\begin{lemma} \label{Lemma compact state}The state space can be chosen to be compact: There exists a compact set $\bar{X} =[0 ,\bar{x}]$ such that $Q (x ,s ,\bar{X}) =1$ whenever $x \in \bar{X}$ and all $s \in \mathcal{P} (X)$.

\end{lemma}
\begin{proof}
Fix $s \in \mathcal{P} (X)$. Since $\max \Gamma  (x) =(1 -\delta ) x +k (\bar{a})$, for all $x >0$, we have
\begin{equation*}\frac{\mu  (x ,s) \zeta _{n}}{x} \leq (1 -\delta ) \zeta _{n} +\frac{k (\bar{a}) \zeta _{n}}{x}.
\end{equation*}The last inequality
and the fact that $(1 -\delta ) \zeta _{n} <1$ imply that there exists $\bar{x}$ (that does not depend on $s$) such that $\mu  (x ,s) \zeta _{n} <x$ for all $x \geq \bar{x}$. 

Let $\bar{X} =[0 ,\bar{x}]$. For all $s \in \mathcal{P} (X)$ and $\zeta  \in E$, if $x \in \bar{X}$ we have
\begin{equation*}\mu  (x ,s) \zeta  \leq \mu  (\bar{x} ,s) \zeta _{n} <\bar{x}.
\end{equation*}
That is, $\mu  (x ,s) \zeta  \in \bar{X}$. Thus, $Q (x ,s ,\bar{X}) =\Pr  (\mu  (x ,s) \zeta  \in \bar{X}) =1$ whenever $x \in \bar{X}$. 
\end{proof}

\begin{lemma}
\label{Lemma Q is ergodic}$Q$ is $\bar{X}$-ergodic. 
\end{lemma}

\begin{proof}
Fix $s \in \mathcal{P} (X)$. Define the sequences $x_{k +1} =\mu  (x_{k} ,s) \zeta _{n}$ and $y_{k +1} =\mu  (y_{k} ,s) \zeta _{1}$ where $x_{1} =0$ and $y_{1} =\bar{x}$. Note that $\{x_{n}\}_{n =1}^{\infty }$ is strictly increasing, i.e., $x_{k +1} >x_{k}$ for all $k$. To see this, first note that $x_{2} =\mu  (x_{1} ,s) \zeta _{n} \geq k (0) \zeta _{n} >0 =x_{1}$. Now if $x_{k} >x_{k -1}$, then $\mu $ being strictly increasing in $x$ (see Lemma \ref{Lemma properties} part (i)) implies that $x_{k +1} =\mu  (x_{k} ,s) \zeta _{n} >\mu  (x_{k -1} ,s) \zeta _{n} =x_{k}$. Let $C_{s} =\min \{x \in \mathbb{R}_{ +} :\mu  (x ,s) \zeta _{n} =x\}$. From the facts that $\mu  (0 ,s) \zeta _{n} \geq k (0) \zeta _{n} >0$, $\mu  (\bar{x} ,s) \zeta _{n} <\bar{x}$ (see Lemma \ref{Lemma compact state}), and $\mu $ is continuous (see Lemma \ref{Lemma 1}), by Brouwer fixed point theorem $C_{s}$ is well defined. Similarly, the sequence $\{y_{n}\}_{n =1}^{\infty }$ is strictly decreasing and therefore converges to a limit $C_{s}^{ \ast }$. 

We claim that $C_{s} > C_{s}^{ \ast }$. To see this, first note that Lemma \ref{Lemma compact state} implies that the function $f_{s}$, defined by $f_{s} (x,\zeta ) =\mu  (x ,s) \zeta $, is from $\bar{X} \times E$ into $\bar{X}$. Note that $f_{s}$ is increasing in both arguments and that $\bar{X}$ is a complete lattice. Thus, Corollary 2.5.2 in \cite{topkis2011supermodularity} implies that the greatest and least fixed points of $f_{s}$ are increasing in $\zeta $. Lemma \ref{Lemma properties} part (ii) and $\zeta _{1}<1$ imply that $f_{s} (x ,\zeta _{1}) =\mu  (x ,s) \zeta _{1}$ is a contraction mapping from $\bar{X}$ to itself. Thus, $f_{s} (x ,\zeta _{1})$ has a unique fixed point which equals the limit of the sequence $\{y_{n}\}_{n =1}^{\infty }$, $C_{s}^{\ast}$. Since the least fixed point of $f_{s}$ is increasing in $\zeta$ we conclude that $C_{s} \geq C_{s}^{ \ast }$. Since $\mu $ is increasing and positive we have $C_{s} =\mu (C_{s} ,s)\zeta _{n} >\mu (C_{s} ,s)\zeta _{1} \geq \mu (C_{s}^{ \ast } ,s)\zeta _{1} =C_{s}^{ \ast }$ 

Let $x^{ \ast } =(C_{s} +C_{s}^{ \ast })/2$. Since $x_{k}\uparrow C_{s}$ and $y_{k}\downarrow C_{s}^{ \ast }$, there exists a finite $N_{1}$ such that $x_{k} >x^{ \ast }$ for all $k \geq N_{1}$, and similarly, there exists a finite $N_{2}$ such that $y_{k} <x^{ \ast }$ for all $k \geq N_{2}$. Let $m =\max \{N_{1} ,N_{2}\}$. Thus, after $m$ periods there exists a positive probability ($\zeta _{1}^{m}$) to move from the state $\bar{x}$ to the set $[0 ,x^{ \ast }]$, and a positive probability to move from the state $0$ to the set $[x^{ \ast } ,\bar{x}]$. That is, we found $x^{ \ast } \in [0 ,\bar{x}]$ and $m >0$ such that $Q^{m}(\bar{x} ,s ,[0 ,x^{ \ast }]) >0$ and $Q^{m}(0 ,s ,[x^{ \ast } ,\bar{x}]) >0$. Since $\bar{X}$ is compact and $Q$ is increasing in $x$, then $Q$ is $\bar{X}$-ergodic (see Theorem 2 in \cite{hopenhayn1992} or Theorem 2.1 in \cite{bhattacharya1988asymptotics}). 
\end{proof}

\vspace{5mm} 

Now, we prove Theorem \ref{Theorem DOP comp}. The main idea behind the proof is to show that the optimal stationary strategy $g$ is increasing or decreasing in the relevant parameter using a lattice-theoretical approach and then to conclude that the conditions of Theorem \ref{Theorem comp} hold. 

Let $(I , \succeq_{I} )$ be a partial order set that influences the firms' decisions. We denote a generic element in $I\;$ by $e$. For instance, $e$ can be the discount factor or the cost of a unit of investment. Throughout the proof of Theorem \ref{Theorem DOP comp}
we allow an additional argument in the functions that we consider. For instance, the value function $V$ is denoted by:
\begin{equation*}V (x ,s ,e) =\underset{a \in [0 ,\bar{a}]}{\max } ~h (x ,a ,s ,e ,V).
\end{equation*} Likewise, the optimal stationary strategy is denoted by $g (x ,s ,e)$, and $u (x ,s ,e)$ is the one-period profit function. Here, we come back to the original formulation over actions $a$.

\begin{proof}
[Proof of Theorem~\ref{Theorem DOP comp}] i) Assume that $f \in B (X \times \mathcal{P} (X) \times I)$ is concave in the first argument and has decreasing differences in $(x ,d)$ where $I \subseteq \mathbb{R}_{ +}$ is the set of all possible unit investment costs endowed with the natural order, $d_{2} \geq d_{1}$. 

Fix $s \in \mathcal{P} (X)$. Note that $da$ has increasing differences in $(a,d)$. Thus, $u (x ,s) -d a$ has decreasing differences in $(a ,d)$, $(x ,a)$ and $(x ,d)$. Since $f$ has decreasing differences and $k$ is increasing, the function $f (((1 -\delta ) x +k (a)) \zeta  ,s ,d)$ has decreasing differences in $(a ,d)$ and $(x ,d)$ for every $\zeta  \in E$. Since $f$ is concave in the first argument and $k$ is increasing, it can be shown that the function $f (((1 -\delta ) x +k (a)) \zeta  ,s ,d)$ has decreasing differences in $(x ,a)$ for every $\zeta  \in E$. Thus, the function
\begin{equation*}h (x ,a ,s ,d ,f) =u (x ,s) -d a +\beta  \sum _{j =1}^{n}p_{j} f (((1 -\delta ) x +k (a)) \zeta _{j} ,s ,d)
\end{equation*}has decreasing differences in $(x ,a)$, $(x ,d)$ and $(a ,d)$ as the sum of functions with decreasing differences. 

A similar argument to Lemma 1 in \cite{hopenhayn1992} or Lemma 2 in \cite{lovejoy1987ordered}
implies that if $h (x ,a ,s ,d ,f)$ has decreasing differences in $(x ,a)$, $(x ,d)$ and $(a ,d)$, then $T f (x ,s ,d) =\underset{a \in [0 ,\bar{a}]}{\max }h (x ,a ,s ,d ,f)$ has decreasing differences in $(x ,d)$. The proof of Lemma \ref{Lemma single-valued} implies that
$Tf$ is concave in $x$. We conclude that for all $n =1 ,2 ,3...$, $T^{n} f$ is concave in $x$ and has decreasing differences. Standard dynamic programming arguments (see \cite{bertsekas1978stochastic}) show that $T^{n} f$ converges to $V$ uniformly. Since the set of functions with decreasing differences is closed under uniform convergence, $V$ has decreasing differences in $(x ,d)$. From the same argument as above, $h (x ,a ,s ,d ,V)$ has decreasing differences in $(a ,d)$. Theorem 6.1 in \cite{topkis1978} implies that $g (x ,s ,d)$ is decreasing in $d$. 

Define the order $ \succeq _{I}$ by $d_{2}  \succeq _{I}d_{1}$ if and only if $d_{1} \geq d_{2}$. Thus $d_{2}  \succeq _{I}d_{1}$ implies that
\begin{align*}Q (x ,s ,d_{2} ,B) &  =\Pr (((1 -\delta ) x +k (g (x ,s ,d_{2})) \zeta  \in B) \\
 &  \geq \Pr (((1 -\delta ) x +k (g (x ,s ,d_{1})) \zeta  \in B) \\
 &  =Q (x ,s ,d_{1} ,B)\end{align*}for all $x ,s$ and every upper set $B$, because $g (x ,s ,d)$ is decreasing in $d$. That is, $Q (x ,s ,d_{2} , \cdot ) \succeq  _{SD}Q (x ,s ,d_{1} , \cdot )$ for all $x ,s$ and $d_{2} ,d_{1} \in I$ such that $d_{2}  \succeq _{I}d_{1}$. From Theorem \ref{Theorem comp} and Theorem \ref{Unique DOP} we conclude that $d_{2}  \succeq _{I}d_{1}$ implies $s (d_{2}) \succeq s (d_{1})$, i.e., $d_{2} \leq d_{1}$ implies $s (d_{2}) \succeq s (d_{1})$. 

(ii) The proof of part (ii) is the same as the proof of part (i)
and is therefore omitted. 

(iii) Assume that $f \in B (X \times \mathcal{P} (X) \times I)$ is increasing in the first argument and has decreasing differences in $(x ,\beta )$ where $I =(0 ,1)$ is the set of all possible discount factors endowed with the reverse order; $\beta _{2}  \succeq _{I}\beta _{1}$ if and only if $\beta _{1} \geq \beta _{2}$. A standard argument shows that $T f$ is increasing in the first argument. We will only show that $h (x ,a ,s ,\beta  ,f)$ has decreasing differences in $(a ,\beta )$ and $(x ,\beta )$; the rest of the proof is the same as the proof of part (i). Fix $s$, $x$ and let $\beta _{2}  \succeq _{I}\beta _{1}$ (i.e., $\beta _{1} \geq \beta _{2}$), and $a_{2} \geq a_{1}$. Decreasing differences of $f$ and the fact that $k$ is increasing imply that $f (((1 -\delta ) x +k (a_{2})) \zeta  ,s ,\beta ) -f (((1 -\delta ) +k (a_{1})) \zeta  ,s ,\beta )$ is decreasing in $\beta $ for all $\zeta  \in E$. Since $\beta _{1} \geq \beta _{2}$, $f$ and $k$ are increasing, we have
\begin{align*}\beta _{2} \sum _{j =1}^{n}p_{j} (f (((1 -\delta ) x +k (a_{2})) \zeta _{j} ,s ,\beta _{2}) -f (((1 -\delta ) x +k (a_{1})) \zeta _{j} ,s ,\beta _{2})) \\
 \leq \beta _{1} \sum _{j =1}^{n}p_{j} (f (((1 -\delta ) x +k (a_{2})) \zeta _{j} ,s ,\beta _{1}) -f (((1 -\delta ) x +k (a_{1})) \zeta _{j} ,s ,\beta _{1})).\end{align*}
 Thus $h (x ,a ,s ,\beta  ,f)$ has decreasing differences in $(a ,\beta )$. A similar argument shows that $h (x ,a ,s ,\beta  ,f)$ has decreasing differences in $(x ,\beta )$. 
\end{proof}

\vspace{5mm}

\begin{proof}[Proof of Theorem~\ref{Theorem:dynamic advertising}]
(i) The proof of the Theorem is similar to the proof of Theorem \ref{Unique DOP}.
The idea of the proof is to show that the
conditions of Theorem \ref{Theorem (Coupling actions)} hold. We now show that $\overline{Q}$ is increasing in $x$ and decreasing in $s$ (see Section \ref{Sec: Coupling th} for the definition of $\overline{Q}$). 

We use the same change of variables and notation as in the proof of Theorem \ref{Unique DOP}. Define
\begin{equation}z =(1 -\delta )(x +a)
\end{equation}
and note that $a =(1 -\delta )^{-1}z -x$. The resulting Bellman operator is given by
\begin{equation*}Kf(x ,s) =\underset{z \in \Gamma  (x)}{\max }J(x,z ,s ,f),
\end{equation*}
where $\Gamma  (x) =[(1 -\delta )(x +1) ,(1 -\delta )(x +\overline{a})]$,
\begin{equation*}J (x ,z ,s ,f) =\pi  (x ,z ,s) +\beta  \sum _{j}p_{j} f\left (z\zeta _{j} ,s\right ),
\end{equation*} 
and 
\begin{align*}\pi (x ,z ,s) &  =r\frac{(x +(1 -\delta )^{-1}z -x)^{\gamma _{1}}}{\left (\int (x^{ \prime } +(1 -\delta )^{ -1}z^{\prime } -x^{ \prime })s(dx^{ \prime},dz^{ \prime })\right )^{\gamma _{2}}} -(1 -\delta )^{ -1}z -x \\
 &  =r\frac{\left ((1 -\delta )^{ -1}z\right )^{\gamma _{1}}}{\left (\int (1 -\delta )^{ -1}z^{ \prime }s(dx^{ \prime } ,dz^{ \prime })\right )^{\gamma _{2}}} -x -(1 -\delta )^{ -1}z.\end{align*}
 Let $\mu  (x ,s) =\operatorname{argmax}_{z \in \Gamma  (x)}J (x ,z ,s ,V)$. Since $\pi$ is concave in $(x,z)$, Lemma \ref{Lemma single-valued} implies that the policy function $\mu (x ,s)$ is single-valued.

It is immediate that $\pi $ has increasing differences in $(x ,z)$, and decreasing differences in $(z ,s)$ and $(x ,s)$. Here $s_{2} \succeq s_{1}$ if and only if 
$$\int (1 -\delta )^{ -1}z^{ \prime }s_{2}(dx^{ \prime } ,dz^{ \prime }) \geq \int (1 -\delta )^{ -1}z^{ \prime }s_{1}(dx^{ \prime } ,dz^{ \prime }).$$ 
From Lemma \ref{Lemma single-valued}, we can show that $\mu $ is increasing in $x$ and decreasing in $s$. 

Thus, for each $s \in \mathcal{P} (X \times A)$, $x_{2} \geq x_{1}$ and any upper set $B \times D \in \mathcal{B}(X \times A)$ we have 
\begin{align*}\overline{Q}(x_{2} ,s ,B \times D) &  =\sum _{j =1}^{n}p_{j}1_{B \times D}(\mu (x_{2} ,s)\zeta _{j} ,\mu \left (\mu (x_{2} ,s)\zeta _{j} ,s\right )) \\
 &  \geq \sum _{j =1}^{n}p_{j}1_{B \times D}(\mu (x_{1} ,s)\zeta _{j} ,\mu (\mu (x_{1} ,s)\zeta _{j} ,s)) \\
 &  =\overline{Q}(x_{1} ,s ,B \times D) .\end{align*}
 The equalities follow from the proof of Theorem \ref{Theorem (Coupling actions)}. The inequality follows because $\mu $ is increasing in $x$. Thus, $\overline{Q} (x_{2} ,s , \cdot ) \succeq  _{SD}\overline{Q} (x_{1} ,s , \cdot )$, i.e., $\overline{Q}$ is increasing in $x$.   
 
Similarly, because $\mu  (x ,s)$ is decreasing in $s$, we can show that $\overline{Q}$ is decreasing in $s$ for each $x \in X$.  

We conclude that $\overline{Q}$ is decreasing in $s$ and increasing in $x$. Compactness of the state space $X$ and $X$-ergodicity of $\overline{Q}$ can be established using similar arguments to the arguments in Theorem \ref{Unique DOP}.  Thus, all the conditions of Theorem \ref{Theorem (Coupling actions)} parts (i) and (ii) hold. We conclude that the dynamic advertising model has a unique MFE. 

 The proofs of parts (ii) and (iii) are similar to the proof of Theorem \ref{Theorem DOP comp} and are therefore omitted.  
\end{proof}

\vspace{5mm}

\begin{proof}[Proof of Theorem~\ref{THM: REPUTATION}]
(i) First note that the state space $X =[0 ,M_{1}] \times [0 ,M_{2}]$ is compact. We now show that $Q$ is increasing in $x_{1}$ and decreasing in $s$ with respect to $x_{1}$.

For the proof of the theorem, it will be convenient to change the decision variable in the Bellman equation.
Define \begin{equation*}z =\frac{x_{2}}{1 +x_{2}}x_{1} +\frac{1}{1 +x_{2}}k (a),
\end{equation*} and note that we can write $a =k^{ -1} (z(1 +x_{2}) -x_{2}x_{1})$, which is well defined because $k$ is strictly increasing. The resulting Bellman operator is given by
\begin{equation*}Kf(x_{1},x_{2} ,s) =\underset{z \in \Gamma  (x_{1},x_{2})}{\max }J (x_{1},x_{2} ,z ,s ,f),
\end{equation*}
where $\Gamma  (x_{1},x_{2}) =[\frac{x_{2}}{1 +x_{2}}x_{1} +\frac{1}{1 +x_{2}}k (0) ,\frac{x_{2}}{1 +x_{2}}x_{1} +\frac{1}{1 +x_{2}}k (\bar{a})]$,
\begin{equation*}J(x_{1},x_{2} ,z ,s ,f) =\pi  (x_{1},x_{2},z ,s) +\beta  \sum _{j}p_{j} f\left (\min \left (z + \frac{\zeta _{j}}{1+x_{2}} ,M_{1} \right ) ,\min (x_{2} +1 ,M_{2}) ,s\right ),
\end{equation*}
and
\begin{equation*}\pi (x_{1},x_{2} ,z ,s) =\frac{\nu (x_{1} ,x_{2})}{\int \nu (x_{1} ,x_{2})s(d(x_{1} ,x_{2}))} -dk^{ -1} (z(1 +x_{2}) -x_{2}x_{1}) .
\end{equation*}
Let $\mu  (x_{1},x_{2} ,s) =\operatorname{argmax}_{z \in \Gamma  (x_{1},x_{2})} J(x_{1},x_{2} ,z ,s ,V)$. From the arguments as the arguments in Lemma \ref{Lemma single-valued}, the optimal stationary strategy $\mu (x_{1},x_{2} ,s)$ is single-valued. 

Let $x^{\prime}_{1} \leq x_{1}$ and $s_{2} \succeq s_{1}$. Because $\nu$ is increasing, we have  
\begin{align*} \nu (x_{1} ,x_{2}) \left (\frac{1}{\int \nu (x_{1} ,x_{2})s_{2}(d(x_{1} ,x_{2}))} - \frac{1}{\int \nu (x_{1} ,x_{2})s_{1}(d(x_{1} ,x_{2}))} \right ) \\ 
\leq \nu (x^{\prime}_{1} ,x_{2}) \left (\frac{1}{\int \nu (x_{1} ,x_{2})s_{2}(d(x_{1} ,x_{2}))} - \frac{1}{\int \nu (x_{1} ,x_{2})s_{1}(d(x_{1} ,x_{2}))} \right ). \end{align*}  
Thus, $\pi$ has decreasing differences in $(x_{1},s)$. In addition, $\pi$ has decreasing differences in $(z,s)$ and increasing differences in $(x_{1},z)$ (see the proof of Lemma \ref{Lemma single-valued}). From Lemma \ref{Lemma single-valued}, we can show that $\mu $ is increasing in $x_{1}$ and decreasing in $s$.

Recall that in every period, with probability $1-\beta$, each seller departs the market and a new seller with state $(0,0)$ immediately arrives to the market. With probability $\beta$, each seller stays in the market  and moves to a new state according to the dynamics described in Section \ref{Sec:dynamic rep}. Thus, we have
\begin{align*}Q(x_{1} ,x_{2} ,s ,B_{1} \times B_{2}) &=(1 -\beta )\delta _{\{(0 ,0)\}}(B_{1} \times B_{2}) \\ &
 +\beta \Pr \left ( (\min \left (\mu (x_{1} ,x_{2} ,s) + \frac{\zeta _{j}}{1+x_{2}}  ,M_{1} \right ) ,\min (x_{2} +1 ,M_{2})) \in B_{1} \times B_{2}\right ) \\ &
 =  (1 -\beta )1_{B_{1} \times B_{2}}(0,0) \\ &
 + \beta \sum _{j=1}^{n}p_{j}1_{B_{1} \times B_{2}}(\min \left (\mu (x_{1} ,x_{2} ,s) + \frac{\zeta _{j}}{1+x_{2}}  ,M_{1} \right ) ,\min (x_{2} +1 ,M_{2}))\end{align*}
 where $\delta _{\{c\}}$ is the Dirac measure on the point $c \in \mathbb{R}^{2}$. 
 Let $f :X_{1} \times X_{2} \rightarrow \mathbb{R}$ be increasing in the first argument. Assume that $x^{\prime}_{1} \leq x_{1}$. We have
\begin{align*}\int _{X}f (y_{1} ,y_{2})Q((x^{\prime}_{1},x_{2}),s,d(y_{1},y_{2})) &  = (1 -\beta )f(0,0) \\
& + \beta \sum _{j=1}^{n}p_{j} f \left (\min \left (\mu (x^{\prime}_{1} ,x_{2} ,s) + \frac{\zeta _{j}}{1+x_{2}} ,M_{1} \right ) ,\min (x_{2} +1 ,M_{2}) \right ) \\ &
\leq (1 -\beta )f(0,0) \\ 
& + \beta \sum _{j=1}^{n}p_{j} f \left (\min \left (\mu (x_{1} ,x_{2} ,s) + \frac{\zeta _{j}}{1+x_{2}} ,M_{1} \right ) ,\min (x_{2} +1 ,M_{2}) \right )   \\ &
=\int _{X}f (y_{1},y_{2})Q((x_{1},x_{2}),s,d(y_{1},y_{2})).
\end{align*} 
The inequality follows from the facts that $\mu $ is increasing in $x_{1}$, and $f$ is increasing in the first argument.

We conclude that $Q$ is increasing in $x_{1}$. Similarly, because $\mu$ is decreasing in $s$, we can prove that $Q$ is decreasing in $s$ with respect to $x_{1}$. We now show that $Q$ is $X$-ergodic. 

The Markov chain $Q$ is said to satisfy the Doeblin condition if there exists a positive integer $n_{0}$, $\epsilon >0$ and a probability measure $\upsilon $ on $X$ such that $Q^{n_{0}}(x ,s,B) \geq \epsilon \upsilon (B)$ for all $x \in X$ and all measurable $B$. From the definition of $Q$, we have $Q(x,s,B) \geq (1-\beta) \delta _{\{(0 ,0)\}}(B)$ for every measurable $B$, so $Q$ satisfies the Doeblin condition. Thus, $Q$ is $X$-ergodic (see Theorem 8 in  \cite{roberts2004general}). 

Thus, all the conditions of Theorem \ref{Corr:product space} and Theorem \ref{Theorem existence} are satisfied. We conclude that the dynamic reputation model has a unique MFE. 

(ii) The proof of part (ii) is similar to the proof of Theorem \ref{Theorem DOP comp} and is therefore omitted.  
\end{proof}

\subsection{Heterogeneous Agent Macro Models: Proof of Corollary \ref{Hetero Macro uniq}}
\begin{proof}
[Proof of Corollary~\ref{Hetero Macro uniq}] From Theorem \ref{Corr:product space} we only need to show that $Q$ is increasing in $x_{1}$ and decreasing in $s$ in order to prove Corollary \ref{Hetero Macro uniq}. 

Let $f :X_{1} \times X_{2} \rightarrow \mathbb{R}$ be increasing in the first argument. Assume that $x^{\prime}_{1} \leq x_{1}$. We have
\begin{align*}\int _{X}f (y_{1} ,y_{2})Q((x^{\prime}_{1},x_{2}),s,d(y_{1},y_{2})) &  =\sum _{j}p_{j} f (\tilde{g} (x^{\prime}_{1} ,x_{2} ,H (s)) ,m (x_{2} ,\zeta _{j})) \\
 &  \leq \sum _{j}p_{j} f (\tilde{g} (x_{1} ,x_{2} ,H (s)) ,m (x_{2} ,\zeta _{j})) \\
 &  =\int _{X}f (y_{1} ,y_{2})Q((x_{1},x_{2}),s,d(y_{1},y_{2})). \end{align*} 
 The inequality follows from the facts that $\tilde{g}$ is increasing in $x_{1}$ and $f$ is increasing in the first argument. In a similar manner, because $\tilde{g}$ is decreasing in the aggregator, we can show that $Q$ is decreasing in $s$ with respect to $x_{1}$.

We conclude that $Q$ is increasing in $x_{1}$ and decreasing in $s$. 
\end{proof}

\subsection{Extensions: Proofs of Theorem \ref{Theorem (Coupling actions)} and Lemma \ref{Lemma: Ex-ante}}

\begin{proof}
[Proof of Theorem~\ref{Theorem (Coupling actions)}]
The proofs of part (i) and of part (iii) are the same as the proofs of Theorem \ref{Theorem uniq} and of Theorem \ref{Theorem comp}. To prove part (ii) we need to show that the operator $\overline{\Phi}:\mathcal{P} (X) \rightarrow \mathcal{P} (X)$ defined by
\begin{equation*}\overline{\Phi } s (B \times D) =_{}\int _{X}\overline{Q}(x ,s ,B \times D)s(dx,A).
\end{equation*}
is continuous (the rest of the proof is the same as the proof of Theorem \ref{Theorem existence}). The continuity of $\overline{\Phi}$ follows from a similar argument to the argument in the proof of Theorem \ref{Theorem existence}. We provide the proof for completeness. 

Note that for every bounded and measurable function $f :X \times A \rightarrow \mathbb{R}$ and for every $s \in \mathcal{P} (X \times A)$ we have
\begin{equation}\int _{X \times A}f\left (x ,a\right )\overline{\Phi }s\left (d\left (x ,a\right )\right ) =\int _{X}\sum \limits _{j =1}^{n}p_{j}f\left (w\left (x ,g(x ,s) ,s ,\zeta _{j}\right ) ,g\left (w\left (x ,g(x ,s),s ,\zeta _{j}\right ) ,s\right )\right )s\left (dx,A\right ). \label{Equation: Action2}
\end{equation}
To see this, first assume that $f =1_{B \times D}$ for some measurable set $B \times D \in \mathcal{B} (X \times A)$.  We have 
\begin{align*}\int _{X \times A}f\left (x ,a\right )\overline{\Phi}s\left (d(x ,a)\right ) &  =\overline{\Phi}s\left (B \times D\right ) \\
 &  =\int _{X}\int _{B}1_{D}(g\left (y ,s\right ))Q\left (x ,s ,dy\right )s(dx ,A) \\
 &  =\int _{X}\int _{X}1_{B}\left (y\right )1_{D}\left (g(y ,s)\right)Q(x ,s ,dy)s(dx ,A) \\
 &  =\int _{X}\int _{X}1_{B \times D}\left (y ,g(y ,s)\right )Q(x ,s ,dy)s(dx ,A) \\
 &  =\int _{X}\sum _{j =1}^{n}p_{j}1_{B \times D}\left (w\left (x ,g\left (x ,s\right ) ,s ,\zeta _{j}\right ) ,g\left (w\left (x ,g\left (x ,s\right ) ,s ,\zeta _{j}\right ) ,s\right )\right )s\left (dx ,A\right) \\
 &  =\int _{X}\sum \limits _{j =1}^{n}p_{j}f\left (w\left (x ,g\left (x ,s\right ) ,s ,\zeta _{j}\right ) ,g\left (w\left (x ,g\left (x ,s\right ) ,s ,\zeta _{j}\right ) ,s\right )\right )s\left (dx ,A\right ). \end{align*} 
 A standard argument shows that Equation (\ref{Equation: Action2}) holds for every bounded and measurable function $f$.

Assume that $s_{n}$ converges weakly to $s$. Thus, the marginal distribution $s_{n}(\cdot,A)$ converges weakly to $s(\cdot,A)$. Let $f :X \times A \rightarrow \mathbb{R}$ be a continuous and bounded function. Because $w$ and $g$ are continuous, we have \begin{equation*}f (w (x_{n} ,g (x_{n} ,s_{n}) ,s_{n} ,\zeta ) ,g(w\left (x_{n} ,g\left (x_{n} ,s_{n}\right ) ,s_{n} ,\zeta \right ) ,s_{n})) \rightarrow f(w (x ,g (x ,s) ,s ,\zeta ) ,g(w\left (x ,g\left (x ,s\right ) ,s ,\zeta \right ),s))
\end{equation*} 
whenever $x_{n} \rightarrow x$. 

Let $$k_{n}\left (x\right ) : =\sum _{j =1}^{n}p_{j}f(w (x ,g (x ,s_{n}) ,s_{n} ,\zeta _{j}) ,g(w\left (x ,g\left (x ,s_{n}\right ) ,s_{n} ,\zeta_{j} \right ) ,s_{n}))$$ 
and 
$$k\left (x\right ) : =\sum _{j =1}^{n}p_{j}f(w (x ,g (x ,s) ,s ,\zeta_{j} ) ,g(w\left (x ,g\left (x ,s\right ) ,s ,\zeta_{j} \right ) ,s)).$$  

Then $k_{n}$ converges continuously to $k$, i.e., $k_{n} (x_{n}) \rightarrow k (x)$ whenever $x_{n} \rightarrow x$. Since $f$ is bounded, the sequence $k_{n}$ is uniformly bounded. Using Proposition \ref{Prop1} yields
\begin{align*}\underset{n \rightarrow \infty }{\lim }\int _{X \times A}f (x,a) \overline{\Phi}  s_{n} (d (x ,a)) &  =\underset{n \rightarrow \infty }{\lim }\int _{X}k_{n} (x) s_{n}(dx,A) \\
 &  =\int _{X}k (x) s (dx,A) \\
 &  =\int _{X \times A}f (x,a) \overline{\Phi}  s (d(x,a))).\end{align*}
 Thus, $\Phi  s_{n}$ converges weakly to $\Phi  s$. We conclude that $\Phi$ is continuous. 
\end{proof}

\begin{proof}[Proof of Lemma \ref{Lemma: Ex-ante}]
Let $f:X \times \Theta \rightarrow R$ be increasing in in the first. The fact that $Q_{\theta}$ is increasing in $x$ implies that the function 
\begin{align*}
    \int _{X \times \Theta} f(y,\theta^{\prime})Q_{h}(x,\theta,s_{h},d(y,\theta ^{\prime}))  =  \int _{X \times \Theta} f(y,\theta^{\prime}) Q_{\theta}(x,s_{h},dy)1_{D}(d\theta ^{\prime})
\end{align*}
is increasing in $x$ for every type $\theta$ and every extended population state $s_{h}$. That is, $Q_{h}$ is increasing in $x$. Similarly, $Q_{h}$ is decreasing in $s_{h}$ with respect to $x$ when $Q_{\theta}$ is decreasing in $s_{h}$. 
\end{proof}

\end{document}